\documentclass[aip, jcp, reprint]{revtex4-1}

\usepackage{enumitem}
\usepackage{graphicx}          
\usepackage{dcolumn}           
\usepackage{bm}                
\usepackage[mathlines]{lineno} 
\usepackage[utf8]{inputenc}
\usepackage{amsmath}    
\usepackage{amssymb}    
\usepackage{amsthm}
\usepackage{graphicx}   
\usepackage{verbatim}   
\usepackage{color}      
\usepackage{subfigure}  
\usepackage{hyperref}   
\usepackage[capitalise]{cleveref}   
\usepackage{appendix}
\DeclareMathOperator*{\argmax}{arg\,max}
\newcommand{\argmin}{\operatornamewithlimits{argmin}}

\DeclareMathOperator{\Tr}{Tr}

\newtheorem{theorem}{Theorem}
\newtheorem{lemma}[theorem]{Lemma}

\begin{document}

\title{Variational cross-validation of slow dynamical modes in molecular kinetics}
\author{Robert T. McGibbon}
\affiliation{Department of Chemistry, Stanford University, Stanford CA 94305, USA}
\author{Vijay S. Pande}
\affiliation{Department of Chemistry, Stanford University, Stanford CA 94305, USA}
\email{pande@stanford.edu}
\date{\today}

\begin{abstract}
  Markov state models (MSMs) are a widely used method for approximating the eigenspectrum of the molecular dynamics propagator, yielding insight into the long-timescale statistical kinetics and slow dynamical modes of biomolecular systems. However, the lack of a unified theoretical framework for
choosing between alternative models has hampered progress, especially for non-experts applying these methods to novel biological systems. Here, we consider cross-validation with a new objective function for estimators of these slow dynamical modes, a generalized matrix Rayleigh quotient (GMRQ), which measures the ability of a rank-$m$ projection operator to capture the slow subspace of the system. It is shown that a variational theorem bounds the GMRQ from above by the sum of the first $m$ eigenvalues of the system's propagator, but that this bound can be violated when the requisite matrix elements are estimated subject to statistical uncertainty. This overfitting can be detected and avoided through cross-validation. These result make it possible to construct Markov state models for protein dynamics in a way that appropriately captures the tradeoff between systematic and statistical errors.
\end{abstract}
\maketitle

\section{Introduction}
Conformational dynamics are central to the biological function of
macromolecular systems such as signaling proteins, enzymes, and
channels. The molecular description of processes as diverse as protein
folding, kinase activation, voltage-gating of ion channels, and
ubiquitin signaling involve not just the structure of a unique single
conformation, but the conformational dynamics between a multitude
of states accessible on the potential energy
surface.\cite{dill2008protein, huse2002conformational,
  vargas2012emerging, phillips2013conformational} These dynamics occur
on a range of timescales and have varying degrees of structural
complexity: localized vibrations may occur on the 0.1 ps timescale,
while large-scale structural changes like protein folding can take
seconds or longer.\cite{careri1975statistical} Although many
experimental techniques -- most notably X-ray crystallography and nuclear magnetic resonance spectroscopy
-- can yield detailed structural information on functional
conformations, the experimental characterization of the dynamical
processes, intermediate conformations and transition pathways in
macromolecular systems remains exceptionally
challenging.\cite{buergi1983crystal, mittermaier2006new}

Atomistic molecular dynamics (MD) simulations can complement
experiment and provide a powerful tool for probing conformational
dynamics, allowing researchers to directly visualize and analyze the
time evolution of macromolecular systems in atomic detail. Three major
challenges for MD simulation of complex systems are the accuracy of
the potential energy functions, adequate sampling of conformational
space, and quantitative analysis of simulation results. The
state-of-the-art on all three fronts has advanced rapidly in recent
years. A new generation of increasingly accurate forcefields have
recently emerged, such as those which include explicit polarizability
and have been parameterized more
systematically.\cite{wang2012systematic, huang2013automated,
  ponder2010current, best2012optimization, lopes2013polarizable} On
the sampling problem, the introduction of graphical processing units
(GPUs) has dramatically expanded the timescales accessible with MD
simulation at modest cost, and specialized MD-specific hardware and
distributed computing networks have yielded further
gains.\cite{stone2010gpu, shaw2008anton, shirts2000science,
  buch2010high, kohlhoff2014cloud} In this work, we focus on the
remaining challenge, the quantitative analysis of MD simulations.

Despite, or perhaps because of their detail, MD simulations require
further analysis in order to yield insight into macromolecular
dynamics or quantitative predictions capable of being tested
experimentally. The direct result of a simulation, an MD trajectory,
is a time series of Cartesian positions (and perhaps momenta) of
dimension $3N$ ($6N$ if momenta are retained), where
$N$ is the number of atoms in the system. Because routine MD
simulations may contain tens or hundreds of thousands of atoms, these
time series are extremely high-dimensional. A multitude of methods
have been proposed for reducing the dimensionality or complexity of MD
trajectories and enabling the analysis of the system's key long-lived
conformational states, dynamical modes, transition pathways, and
essential degrees of freedom.\cite{chodera2007automatic,
  deuflhard2005robust, rohrdanz2011determination, altis2007dihedral,
  das2006low, krivov2004hidden, weinan2006towards}


In this work, we combine two central ideas from machine learning and chemical
physics -- hyperparameter selection via cross-validation and variational
approaches for linear operator eigenproblems -- to create a new method for
discriminating between alternative simplified models for molecular kinetics
constructed from MD simulations. Towards this end, we prove a new variational
theorem concerning the simultaneous approximation of the first $m$
eigenfunctions of the propagator of a high-dimensional reversible stochastic
dynamical system, which mathematically formalize the slow collective dynamical
motions we wish to identify in molecular systems.

\section{Cross Validation}
\label{sect:2}

In seeking to estimate the slowest collective dynamical modes of a molecular system from a finite set of stochastic trajectories, statistical noise is unavoidable. These dynamical modes, which we formally identify as the first $m$ eigenfunctions, $\phi_i$, of the propagator, an integral operator associated with the dynamics of a molecular system (vide infra), are functions on $\mathbb{R}^{3N}$ to $\mathbb{R}$. Like the ground state wave function in quantum chemistry, these eigenfunctions can only be approximately represented in any finite basis set. Reducing this approximation error, a statistical bias, motivates the use of larger and more flexible basis sets. Unfortunately, in an effect known as the \emph{bias-variance tradeoff},\cite{McGibbon2014Statistical, vapnik1998statistical} larger basis sets tend to exacerbated a competing source of error, the model variance: with a fixed simulation data set but additional adjustable parameters due to a larger basis set, the model estimates of these eigenfunctions become more unstable and uncertain.


As has been known since at least the early 1930s, training a
statistical algorithm and evaluating its performance on the same data set
generally yields overly optimistic results.\cite{larson1931shrinkage} For this reason, standard practice in supervised machine learning is to divide a data set into separate training and test sets. The model parameters are estimated using the training data set, but performance is evaluated separately by scoring the now-trained model on the separate test set, consisting of data points that were left out during the training phase. To avoid overfitting, the choice between alternative models is made using test set performance, not training set performance.

However, because researchers may expend significant effort to collect data sets, the exclusion of a large fraction of the of the data set from the training phase can be a costly proposition. $k$-fold cross-validation is one alternative that can be more data-economical, where the data is split into $k$
disjoint subsets, each of which is cycled as both the training and test set.

Let $X$ be a collection of molecular dynamics trajectories (the data set), which we assume for simplicity to consist of a multiple of $k$ independent MD trajectories of equal length. In $k$-fold cross validation, the trajectories are split into $k$ equally-sized disjoint subsets, called folds, denoted $X^{(i)}$, for $i \in \{1, 2, \ldots, k\}$. These will serve as the test sets. Let $X^{(-i)} = X \setminus X^{(i)}$ denote the set of all trajectories excluded from fold $i$; these will serve as the training sets.

Consider an algorithm to estimate the $m$ slowest dynamical modes of the system, $g$. Examples of such estimators include Markov state models (MSMs)\cite{prinz2011markov} and time-structured independent components analysis (tICA).\cite{schwantes2013Improvements, perezhernandez2013Identification} The result of this calculation, the estimated eigenfunctions, $\hat{\phi}_{1:m}$, are taken to be a function of both an input dataset, $X$, as well as a set of hyperparameters, $\theta$, which many include settings such as the number of states or clustering algorithm in an MSM or the basis set used in tICA.

\begin{align}
\hat{\phi}_{1:m} = (\hat{\phi}_{1}, \hat{\phi}_{2}, \ldots, \hat{\phi}_{m}) = g(X, \theta)
\end{align}

Furthermore, consider an objective function, $O(\hat{\phi}_{1:m}, X')$, which evaluates a set of proposed eigenfunctions, $\hat{\phi}_{1:m}$, and a (possibly new) dataset $X'$, returning a single scalar measuring the performance or accuracy of these eigenfunctions. The mean cross validation performance of a set of hyperparameters is defined by the following expression, which builds $k$ models on each of the training sets and scores them on the corresponding test sets.

\begin{align}
MCV(\theta) = \frac{1}{k}\sum_{i=1}^k O(g(X^{(-i)}, \theta),\, X^{(i)})
\end{align}

Model selection can be performed by finding the hyperparameters, $\theta^* = \argmax_\theta MVC(\theta)$, which maximize the cross validation performance. Many variants of this protocol with different procedures for splitting the training and test sets, such as repeated random subsampling cross-validation and leave-one-out cross validation, are also possible.\cite{maimon2005data}

The remainder of this work seeks to develop a suitable objective function, $O$, for estimates of the slow dynamical modes in molecular kinetics that can be used as shown above in a cross-validation protocol. This task is complicated by the fact that no ground-truth values of true eigenfunctions, $\phi_i$, are available, either in the training or test sets. Nevertheless, our goal is to construct a \emph{consistent} objective function, such that as the size of a molecular dynamics data set, $X$ grows larger, the maximizer of $O(\cdot, X)$ converges in probability to the true propagator eigenfunctions, $\phi_{1:m}$.

\begin{align}
\label{eq:consistent}
\argmax_{\hat{\phi}_{1:m}}\; O(\hat{\phi}_{1:m}, X) \xrightarrow{p} \phi_{1:m}
\end{align}

\section{Theory Background}

We begin by introducing the notion of the propagator and its eigenfunctions from a mathematical
perspective, introducing the key variables and notation that will be
essential for the remainder of this work. We largely follow the order
of presentation in \citet{prinz2011markov} which contains a longer and
more thorough discussion.

Consider a time-homogeneous, ergodic, continuous-time Markov process
$\mathbf{x}(t) \in \Omega$ which is reversible with respect to a
stationary distribution stationary distribution $\mu(\mathbf{x}) :
\Omega \rightarrow \mathbb{R}^+$. Where necessary for concreteness, we take the phase space, $\Omega$, to be $\mathbb{R}^{3N}$, where $N$ is the number of atoms of a molecular system. The system's stochastic evolution over an
interval $\tau > 0$ is described by a transition probability density
\begin{equation}
p(\mathbf{x},\mathbf{y}; \tau) d\mathbf{y} = \mathbb{P}[\mathbf{x}(t + \tau) \in B_\epsilon(\mathbf{y}) \;|\; \mathbf{x}(t) = \mathbf{x}],
\end{equation}
where $B_\epsilon(\mathbf{y})$ is the open $\epsilon$-ball centered at $\mathbf{y}$ with infinitesimal measure $d\mathbf{y}$.

Consider an ensemble of such systems at time $t$, distributed
according to some probability distribution $p_t(\mathbf{x})$. After
waiting for a duration $\tau$, the distribution evolves to a new distribution,
\begin{equation}
\label{eq:prop}
p_{t+\tau}(\mathbf{y}) = \int_\Omega d\mathbf{x} \; p(\mathbf{x}, \mathbf{y}; \tau) \, p_t(\mathbf{x}) = \mathcal{P}(\tau) \circ p_t(\mathbf{y}),
\end{equation}
which defines a continuous integral operator, $\mathcal{P}(\tau)$,
called the propagator with lag time $\tau$. The propagator,
$\mathcal{P}(\tau)$, admits a natural decomposition in terms of
its eigenfunctions and eigenvalues
\begin{equation}
\label{eq:eig}
\mathcal{P}(\tau) \circ \phi_i = \lambda_i \phi_i.
\end{equation}
Furthermore, due to the reversibility of the underlying dynamics, $\mathcal{P}(\tau)$ is compact and self-adjoint with respect the
a $\mu^{-1}$ weighted scalar product,\cite{schutte98conformational}
\begin{align}
\langle f, g \rangle_{\mu^{-1}} = \int_\Omega
d\mathbf{x}\; f(\mathbf{x}) g(\mathbf{x}) \mu^{-1}(\mathbf{x}),
\end{align}
where $f$ and $g$ are arbitrary scalar functions on $\Omega$. The propagator has a unique largest eigenvalue $\lambda_1=1$ with corresponding
eigenfunction $\phi_1(\mathbf{x}) = \mu(\mathbf{x})$. The remaining
eigenvalues are real and positive, can be sorted in descending
order, and can be normalized to be $\mu^{-1}$-orthonormal. Using the spectral decomposition of $\mathcal{P}(\tau)$, the conformational distribution of an ensemble at arbitrary multiples of $\tau$ can be
written as a sum of exponentially decaying relaxation processes
\begin{align}
\label{eq:expansion}
p_{t+k\tau}(\mathbf{x}) &= \sum_{i=1}^\infty \lambda_i^k \langle p_t, \phi_i \rangle_{\mu^{-1}} \phi_i, \\
&= \mu(\mathbf{x}) + \sum_{i=2}^\infty \exp \left(-\frac{k\tau}{t_i}\right) \langle p_t, \phi_i \rangle_{\mu^{-1}} \phi_i,
\end{align}
where $t_i = -\displaystyle\frac{\tau}{\ln \lambda_i}$. The
eigenfunctions $\phi_i(\mathbf{x})$ for $i = 2, ...$ can thus be
interpreted as dynamical modes, along each of which the system relaxes towards equilibrium with a characteristic timescale, $\tau_i$. Many molecular systems are characterized by $m$ individual \emph{slow} timescales with eigenvalues close to one, separated from the remaining eigenvalues by a \emph{spectral
  gap}. These slowest collective degrees of freedom, such as protein folding coordinates and pathways associated with enzyme activation/deactivation, are often identified with key functions in biological systems. The remaining small eigenvalues correspond to faster dynamical processes that rapidly decay to equilibrium. Under these conditions, the long-time dynamics induced by the propagator can be well described by consideration of only these slow eigenfunctions -- that is, a rank-$m$ low-rank approximation to the propagator.

Furthermore, not only do these slow eigenfunctions form a convenient basis, in fact they lead to an \emph{optimal} reduced-rank description of the dynamics. That is, each of the partial sums formed by truncating the expansion in \cref{eq:expansion} at its first $m$ terms is is the closest possible rank-$m$ approximation to $\mathcal{P}(\tau)$ in spectral norm. This statement is made precise by the following theorem.

\begin{theorem}
\label{thm:1}
Let $\mathcal{P}$ be compact linear operator which is self-adjoint with respect to an inner product $\langle \cdot, \cdot \rangle_{\mu^{-1}}$. Assume that the eigenvalues $\lambda_i$ and associated eigenfunctions $\phi_i$ are sorted in descending order by eigenvalue. Define the rank-$m$ operator $\mathcal{P}_m$ such that $ \mathcal{P}_m \circ f = \sum_{i=1}^m \lambda_i \langle f, \phi_i \rangle_{\mu^{-1}} \phi_i$, and let $\mathcal{A}_m$ be an arbitrary rank $m$ operator. Then,
\begin{equation}
\mathcal{P}_m = \argmin_{\mathrm{rank}(\mathcal{A}_m) \leq m} ||\mathcal{A}_m - \mathcal{P}||_{\mu^{-1}}.
\end{equation}
\end{theorem}
\begin{proof}
This is the extension of the familiar Eckart-Young theorem to self-adjoint linear operators. The original result is by Schmidt.\cite{schmidt1907zur} See Courant and Hilbert (pp 161),\cite{courant2008methods} and \citet{micchelli1971some} for further details.
\end{proof}
While \cref{thm:1} is a statement about operator approximation, it can also be viewed as a statement about optimal dimensionality reduction for description of slow dynamics. Over all $m$-dimensional dimensionality reductions, the one which projects the dynamics onto its first $m$ propagator eigenfunctions preserves to the largest degree information about the long-timescale evolution of the original system.

Note however that rank-constrained propagator, $\mathcal{P}_m$, while optimal by spectral norm is not generally positivity-preserving, as proved in \cref{sect:tension}, which is an important property of the propagator necessary for its probabilistic interpretation in \cref{eq:prop}.

\section{Objective function and subspace variational principle}
In this section we introduce the objective function discussed abstractly in \cref{sect:2}. We show that both the existing time-structure independent components analysis (tICA)\cite{schwantes2013Improvements, perezhernandez2013Identification} and Markov state model (MSM)\cite{noe2008transition, bowman2009progress, pande2010everything, prinz2011markov, chodera2014markov} methods can be interpreted as procedures which directly optimize this criteria using different restricted families of basis functions. Furthermore, we show that in the infinite-data limit, when the requisite matrix elements can be computed without error, a variational bound governs this objective function: \emph{ansatz} eigenfunctions, $\hat{\phi}_{1:m}$, which differ from the propagator's true eigenfunctions, $\phi_{1:m}$, are always assigned a score which is less than the score of the true eigenfunctions.

Unfortunately, in the more typical finite-data regime, this variational bound can be violated in a pernicious manner: as the size of the basis set increases past some threshold, models can give continually-increasing training set performance (even breaking the variational bound), even as they get \emph{less} accurate when measured on independent test data sets. This observation underscores the practical value of cross-validation in estimators for the slow dynamical processes in molecular kinetics.

Our results build on the important contributions of \citet{noe2013variational} and \citet{nuske2014variational}, who introduced a closely related variational approach for characterizing the slow dynamics in molecular systems. Our novel contribution stems from an approach to the problem through the lens of cross-validation, with its need for a single scalar objective function. While previous work focuses on estimators of each of the propagator eigenfunctions, $\phi_i$, one at at time, we focus instead on measures related to the simultaneous estimation of all of the first $m$ eigenfunctions collectively.

\begin{theorem}
\label{thm:2}
Let $\mathcal{P}$ be compact linear operator whose eigenvalues $\lambda_1 > \lambda_2 \geq \lambda_3, \dots$ are bounded from above and which is self-adjoint with respect to an inner product $\langle \cdot, \cdot \rangle_{\mu^{-1}}$. Furthermore, let $f$ be an arbitrary set of $m$ linearly independent functions on $\Omega \rightarrow \mathbb{R}$, $f=\{f_i(\cdot)\}_{i=1}^m$. Let $\mathbb{S}^{m}$ and $\mathbb{S}^{m}_{++}$ be the space of $m \times m$ real symmetric matrices and positive definite matrices respectively. Define a matrix $P \in \mathbb{S}^{m}$ with $P_{ij} = \langle f_i, \mathcal{P} \circ f_j \rangle_{\mu^{-1}}$, and a matrix $Q \in \mathbb{S}^{m}_{++}$ with $Q_{ij} = \langle f_i, f_j \rangle_{\mu^{-1}}$. Define $\mathcal{R}_\mathcal{P}[f]$ as
\begin{equation}
\mathcal{R}_\mathcal{P}[f] = \Tr\big( PQ^{-1}\big).
\end{equation}
Then,
\begin{equation}
\label{th:varspace}
\mathcal{R}_\mathcal{P}[f] \leq \sum_{i=1}^m \lambda_i.
\end{equation}
\end{theorem}
\begin{lemma}
\label{lemma}
The equality in \cref{th:varspace} holds for $f = \{\phi_1, \phi_2, \ldots, \phi_m\}$, and any set of $m$ functions, $f$, such that $\mathrm{span}(f) = \mathrm{span}(\{\phi_1, \phi_2, \ldots, \phi_m\})$.
\end{lemma}
The proof of \cref{thm:2} follows from the Ky Fan theorem.\cite{fan1949theorem, overton1992sum} Its proof, as well as the proof of \cref{lemma} can be found in \cref{appendix:proof}.

This result implies that the slow eigenspace of the molecular propagator can be numerically determined by simultaneously varying a set of \emph{ansatz} functions $f$ to maximize $\mathcal{R}_\mathcal{P}[f]$. If the maxima is found, then $f$ are the desired eigenfunctions, up to a rotation. The matrix $P$ has the form of a time-lagged covariance matrix between the \emph{ansatz} functions, describing the tendency of the system to move from regions of phase space strongly associated one \emph{ansatz} function to another in time $\tau$. The matrix $Q$ acts like a normalization, giving the equilibrium overlap between \emph{ansatz} functions. Note that when the trial functions, $f$, are $\mu^{-1}$-orthonormal, $Q$, is simply the identity. Under these conditions, $\mathcal{R_P}[f]$ then assumes a simple form as the sum of the individual Ritz values of the trial functions.

Physically, $\mathcal{R}_\mathcal{P}[f]$ can be interpreted as the ``slowness'' of the lower-dimensional dynamical process formed by projecting a high-dimensional process through the $m$ \emph{ansatz} functions. The maximization of $\mathcal{R}_\mathcal{P}[f]$ corresponds to a search for the coordinates along which the system decorrelates as slowly as possible.

Because it is bounded by the sum of the first $m$ true eigenfunctions of the propagator, $\mathcal{R}_\mathcal{P}[f]$, is the foundation of the sought objective function for cross-validation of estimators for the slow dynamical modes in molecular kinetics. Unfortunately, it cannot be calculated exactly  from a molecular dynamics simulation. Next we show how the requisite matrix elements, $P_{ij}$ and $Q_{ij}$ can be approximated from MD. The noise in these approximations will be a function of both the amount of available simulation data and the size and flexibility of a basis set, leading to the bias variance tradeoff discussed earlier. By the continuous mapping theorem and \cref{thm:2}, consistency of the objective function (in the sense of \cref{eq:consistent}) is established if these estimators for $P$ and $Q$ are consistent.

\subsection{Basis Function Expansion}
Equipped with this variational theorem, we now consider the construction of an approximation to the dominant eigenspace of $\mathcal{P}(\tau)$ using linear combinations of functions from a finite basis set. This reduces the problem of searching over the space of sets of $m$ functions to a problem of finding a weight matrix that linearly mixes the basis functions.

Let $\{\varphi_a\}_{a=1}^n$ be a set of $n$ functions on $\Omega \rightarrow \mathbb{R}$, which will be used as basis functions in which to expand the slowest $m$ propagator eigenfunctions. Physically motived basis functions for protein dynamics might include measurements of protein backbone dihedral angles, the distances between particular atoms, or similar structural coordinates. The basis can also be indicator functions for specific regions of phase space -- the ``Markov states'' in a MSM.

Following \citet{nuske2014variational}, we expand the $m$ \emph{ansatz} eigenfunctions as $\mu$-weighted linear combinations of the basis functions,
$f_i(\mathbf{x}) = \sum_a A_{ia}\mu(\mathbf{x})\varphi_a(\mathbf{x})$, where $A \in \mathbb{R}^{n \times m}$ is a weight matrix of arbitrary expansion coefficients. From the basis functions, we define the time-lagged covariance and overlap matrices $C \in \mathbb{S}^n$ and $S \in \mathbb{S}^n_{++}$ respectively such that $C_{ij} = \langle \mu \varphi_i,\; \mathcal{P} \circ \mu \varphi_j \rangle_{\mu^{-1}}$ and $S_{ij} =\langle \mu \varphi_i,\; \mu \varphi_j \rangle_{\mu^{-1}}$.

Then, by exploiting the linearity of the basis function expansion, the matrices $P$ and $Q$, can be written as matrix products involving the expansion coefficients, correlation and overlap matrices.

\begin{align}
P = A^T C A \\
Q = A^T S A
\end{align}

These equations can be interpreted in a simple way: the time-lagged correlation and overlap of the \emph{ansatz} functions with respect to on another can be formed from two similar matrices involving only the basis functions, $C$ and $S$, and the expansion coefficients, $A$. When the \emph{ansatz} functions, $f$, are constructed this way, $\mathcal{R}_\mathcal{P}[f]$ reduces to the generalized matrix Rayleigh quotient (GMRQ), $\mathcal{R}_\mathcal{P}[f] = \mathcal{R}(A; C, S) = \mathcal{R}(A)$

\begin{align}
\label{eq:ra}
\mathcal{R}(A) \equiv \Tr \big( A^T C A \, (A^T S A)^{-1} \big)
\end{align}

Following \cref{lemma}, we note that $\mathcal{R}(A)$ is a function only of column span of $A$, and is not affected by rescaling, or the application of any invertible transformation of the columns. Therefore, the optimization of $\mathcal{R}(A)$ can be seen as a single optimization problem over the set of all $m$-dimensional linear subspaces of $\Omega$. This space is referred to as a \emph{Grassmann manifold}.\cite{absil2002grassmann} Note that when $m=1$, $P$ and $Q$ are scalars, and $\mathcal{R}(A)$ reduces to a standard generalized Rayleigh quotient.

Furthermore, with fixed basis functions, the training problem, $A^* =
\argmax_A \mathcal{R}(A; C, S)$, is solved directly by a matrix $A^*$
with columns that are the $m$ generalized eigenvectors of $C$ and $S$
with the largest eigenvalues, and this eigenproblem is identical to
the one introduced for the tICA
method,\cite{schwantes2013Improvements,
  perezhernandez2013Identification} and Ritz method.\cite{noe2013variational}

\subsection{Estimation of matrix elements from MD}
\label{sect:msm}
Equipped with a collection of basis functions, $\{\varphi\}$, how can $C$ and $S$ be estimated from an MD dataset? As previously shown by \citet{nuske2014variational}, the matrix elements $C$ and $S$ can be estimated from an equilibrium molecular dynamics simulations, $X = \{\mathbf{x}_t\}_{t=1}^T$, by exploiting the ergodic theorem and measuring the correlation between the basis functions, with or without a time lag.

\begin{align}
C_{ij} &= \langle \mu \varphi_i, \mathcal{P}(\tau) \circ \mu\varphi_j \rangle_{\mu^{-1}} \\
&= \int_{x \in \Omega} \int_{y \in \Omega}\, d\mathbf{x}\; d\mathbf{y}\; \mu(\mathbf{y}) \; \varphi_i(\mathbf{y}) \; p(\mathbf{x}, \mathbf{y}; \tau) \; \varphi_j(\mathbf{x}) \\
&\approx \frac{1}{T-\tau}\sum_{t=1}^{T-\tau} \varphi_i(\mathbf{x}_{t}) \varphi_j(\mathbf{x}_{t+\tau}) \label{eq:chat} \\
S_{ij} &= \langle \mu \varphi_i, \mu \varphi_j \rangle_{\mu^{-1}} \\
&= \int_{x \in \Omega} \varphi_i(\mathbf{x}) \varphi_j(\mathbf{x}) \mu(\mathbf{x}) \\
&\approx \frac{1}{T}\sum_{t=1}^T \varphi_i(\mathbf{x}_t) \varphi_j(\mathbf{x}_t) \label{eq:shat}
\end{align}

Note that for \cref{thm:2} to be applicable, $C$ is required to be symmetric, a property which is likely to be violated by the estimator \cref{eq:chat}. For this reason, in practice we use an estimator that averages the matrix computed in \cref{eq:chat} with its transpose. We call this method transpose symmetrization, and it amounts to including each trajectory twice in the dataset, once in the forward and once in the reversed direction, as discussed in Schwantes and Pande.\cite{schwantes2013Improvements}

Markov state models (MSMs)\cite{noe2008transition, bowman2009progress, pande2010everything, prinz2011markov, chodera2014markov} are particular case of the proposed method that have been widely applied to the analysis of biomolecular simulations\cite{sezer2008using, muff2009identification, buch2011complete, voelz2010molecular, Beauchamp2012simple, zhuang2011simulating,  sadiq2012kinetic, choudhary2014structure, shukla2014activation}, where the basis functions are chosen to indicator functions on a collection of non-overlapping subsets of the conformation space.
Given a set of discrete non-overlapping states which partition $\Omega$, $\mathcal{S} = \{s_i\}_{i=1}^{n}$, such that $s_i \subseteq \Omega$, $\bigcup_{i=1}^{n} s_i = \Omega$, and $s_i \cap s_j = \emptyset$, and define

\begin{align}
\label{eq:MSM}
\varphi_i^\text{MSM}(\mathbf{x}_t) = \begin{cases} 1, &\text{if $\mathbf{x}_t\in s_i$}.\\ 0, &\text{otherwise}.\end{cases}
\end{align}

Using this basis set, as previously shown by \citet{nuske2014variational}, estimates of the correlation matrix elements $C_{ij}$ can be obtained following \cref{eq:chat} by counting the number of observed transitions between sets $s_i$ and $s_j$. The overlap matrix $S$ is diagonal with entries, $S_{ii}$, that estimate the stationary probabilities of the sets, $S_{ii} \approx \int_{\mathbf{x} \in s_i}d\mathbf{x}\, \mu(\mathbf{x})$.

For the particular case of MSM basis sets, in contrast to the somewhat crude
transpose symmetrization method, a more elegant enforcement of symmetry of $C$ can be accomplished via a maximum likelihood estimator following Algorithm 1 of Prinz \emph{et al}.\cite{prinz2011markov}

Equipped with these estimators for $C$ and $S$ from MD data, \cref{eq:ra} now has a form which is suitable for use as a cross-validation objective function, $O(\hat{\phi}_{1:m}, X')$. The proposed eigenfunctions, which may have been trained on a \emph{different} dataset, are numerically represented by expansion coefficients, $\hat{A}$, and $C$ and $S$ act as sufficient statistics from the test dataset $X'$; the GMRQ objective function is $\mathcal{R}(\hat{A}; C(X'), S(X'))$.

\section{Algorithmic Realization}

The central practical purpose of cross-validation with generalized matrix Rayleigh quotient (GMRQ) is, given an MD dataset, to select a set of appropriate basis functions with which to construct Markov state models (MSMs) for system's kinetics. Note that \cref{eq:MSM} leaves substantial flexibility in the definition of the basis set, since the partitioning of phase space into states is left unspecified.

Methodologies for constructing these states include clustering the conformations in the dataset using a variety of distance metrics, clustering algorithms, and dimensionality reduction techniques. Let $\theta$ be a variable
containing the settings for these procedures, which parameterizes a function, $g^\text{MSM}_\theta(X)$, that, given a collection of MD trajectories, returns a set of $n$ states, $\mathcal{S}$.

Procedurally, GMRQ-based cross-validation for MSMs is a protocol for assigning a scalar score, $MCV(\theta)$, to the MSM hyperparameters, $\theta$, with the following steps.
\begin{enumerate}
  \item Separate the full data set into $k$ disjoint folds, as described in \cref{sect:2}.
  \item For each fold, $i$, use the training data set, $X^{(-i)}$, to construct a set of states, $\mathcal{S}^{(-i)} = g_\theta(X^{(-i)})$.
  \item Use the states $\mathcal{S}^{(-i)}$ and the training data set $X^{(-i)}$ to build a Markov state model. This entails clustering the dataset to obtain the basis functions (states), $\{\varphi\}$, estimating the training set correlation and overlap matrices $C^{(-i)}$ and $S^{(-i)}$ from the trajectories, and computing their first $m$ generalized eigenvectors, $A=\argmax_A \mathcal{R}(A; C^{-(i)}, S^{(-i)})$, with a standard generalized symmetric eigensolver (e.g. LAPACK's \textsc{dsygv} subroutine).\cite{lapack99}
  \item These eigenvectors maximize the GMRQ on the training set, but how do they perform when tested on new data? Using the test set data, $X^{(i)}$, and the states, $\mathcal{S}^{(-i)}$, \emph{as determined from the training set}, compute the test set correlation and overlap matrices, $C^{(i)}$ and $S^{(i)}$. These trained eigenvectors, $A$, are scored on the test set by $\mathcal{R}(A;C^{(i)}, S^{(i)})$. The key metric for model selection, the cross-validation mean test set generalized matrix Rayleigh quotient is
  \begin{align}
    MVC(\theta) = k^{-1}\sum_{i=1}^k \mathcal{R}(A;\; C^{(i)}, S^{(i)}).
  \end{align}
 As an overfitting diagnostic, we also calculate the cross-validation mean training set GMRQ,
  \begin{align}
    MVC'(\theta) = k^{-1}\sum_{i=1}^k \mathcal{R}(A; C^{(-i)}, S^{(-i)}).
  \end{align}
  \item Finally, the entire procedure is repeated for many choices of $\theta$,
        and the hyperparameter set that maximize the mean cross validation
        score is chosen as the best model, $\theta^\star = \argmax_\theta MVC(\theta)$.
\end{enumerate}

For this approach, one symptom of overfitting -- the construction of models that describe the statistical noise in $X$ rather than the underlying slow dynamical processes -- is an overestimation of the eigenvalues of the propagator and overestimation of the GMRQ. Related statistical methods, such as kernel principal components analysis which also involve spectral analysis of integral operators under non-negligible statistical error suffer from the same effect, which has been termed variance inflation.\cite{scholkopf1998nonlinear, kjems2000generalizable,
  abrahamsen2011cure}

Left unspecified in this protocol are three important parameters: the degree of cross validation, $k$, the number of desired eigenfunctions, $m$, and the correlation lag time, $\tau$. In our experiments, following common practice in the machine learning community, we use $k=5$. Especially in the data-limited regime, the tradeoffs involving the choice of $k$ are not entirely clear, as the objective lacks the form of an empirical risk minimization problem.\cite{vapnik1998statistical, cornec2010concentration} For MSMs, substantial attention in the literature has been paid to the selection of the lag time, $\tau$.\cite{park2006validation, pande2010everything, prinz2011markov} With fixed basis function, it has been shown that the eigenfunction approximation error is a decreasing function of the $\tau$, which motivates the use of larger values.\cite{sarich2010approximation} On the other hand, larger values of $\tau$ limit the temporal resolution of the model. For MSMs of protein folding, the authors' experience suggest that appropriate values for $\tau$ are often in the range between 1 and 100 nanoseconds. Finally, we suggest that $m$, the rank of GMRQ, be selected based on the number of slow dynamical processes in the system, as determined by an apparent gap in the eigenvalue spectrum of $\mathcal{P}(\tau)$, or heuristically to a value between 2 and $\sim\! 10$.

\section{Simulations}

\subsection{Double Well Potential}
In order to gain intuition about the method, we begin by considering one
of simplest possible systems: Brownian dynamics on a double well
potential. We consider a one dimensional Markov process in which a single
particle evolves according to the stochastic differential equation
\begin{equation}
\frac{dx_t}{dt} = -\nabla V(x_t) + \sqrt{2D}R(t)
\label{eq:sde}
\end{equation}
where $V$ is the reduced potential energy, $D$ is the diffusion
constant, and $R(t)$ is a zero-mean delta-correlated stationary
Gaussian process. For simplicity, we consider the potential
\begin{equation}
V(x) = 1 + \cos(2x)
\end{equation}
with reflecting boundary conditions at $x=-\pi$ and $x=\pi$. Using an
Euler integrator, a time step of $\Delta t = 10^{-3}$, and
diffusion constant $D=10^3$, we simulated 10 trajectories starting
from $x=0$ of length $10^{5}$ steps, and saved the position every 100
steps. The potential and histogram of the resulting data points is
shown in \cref{fig:doublewell} (b). We computed the
true eigenvalues of the system's propagator to machine precision by
discretizing the propagator matrix elements on a dense grid (see \cref{appendix:analytic}). The timescale of the slowest relaxation process in this system is $t_2 \approx 7115.3$ steps, and the dataset contains approximately 94 transitions events.

We now consider the construction of Markov state models for this system, and in
particular the selection of the number of states, $n$, using states, $\mathcal{S}=\{s_i\}_{i=1}^n$, which evenly partition the region between $x=-\pi$ and $x=\pi$ into $n$ equally spaced intervals.

\begin{align}
  \label{eq:si}
  s_i = \Big[-\pi + \frac{2\pi}{n}(i-1),\; -\pi + \frac{2\pi}{n}i \Big)
\end{align}

When $n$ is too
low, we expect that the discretization error in the MSM will dominate,
and our basis will not be flexible enough to capture the first
eigenfunction of the propagator. On the other hand, because the number
of parameters estimated by the MSM is proportional to $n^2$, we expect that for
$n$ too large, our models will be overfit. We therefore use 5-fold cross validation with the GMRQ to select the appropriate number of states, balancing these competing effects. The cross-validation GMRQ for the first
two eigenvectors ($m=2$, $\tau=100$ steps) of the MSMs is shown in \cref{fig:doublewell} (a), along with the exact value of the GMRQ. The blue training curve gives the average GMRQ over the folds when scoring the models on the \emph{same} trajectories that they were fit with, and is simply equal to the mean sum of the first two eigenvalues of the MSMs, whereas the red curve shows the mean GMRQ evaluated on the left-out test trajectories.

\begin{figure*}
\centering
\includegraphics[width=468.0pt]{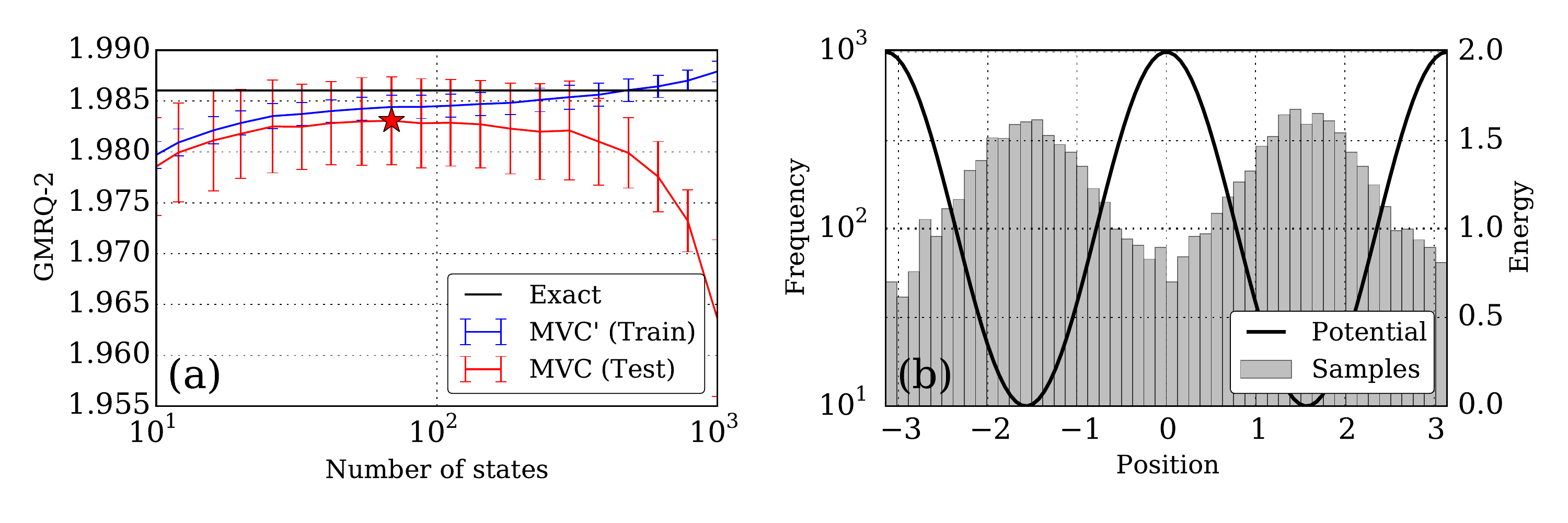}
\caption{\label{fig:doublewell}
Model selection for MSMs of a double well potential. Error bars indicate standard deviations over the 5 folds of cross validation. See text for details.}
\end{figure*}

The training GMRQ increases monotonically, and we note with particular
emphasis that it increases \emph{past} the exact value when using a
large number of states. This indicates that the models built with more
than ~200 states predict \emph{slower} dynamics than the true
propagator. This effect is impossible in the limit of infinite data as demonstrated by \cref{th:varspace} --
it is a direct manifestation of overfitting, and indicates why
straightforward variational optimization without testing on held-out
data or consideration of statistical error fails in a data-limited
regime. The training set eigenvectors, the maximizers of the training set GMRQ, are actually exploiting noise in the dataset more so than modeling the propagator eigenfunctions. On the other hand, the test GMRQ displays an inverted U-shaped behavior and achieves a maximum at $k=61$. These models thus achieve the best \emph{predictive} accuracy in capturing the systems slow dynamics, given the finite data available.


\subsection{Comparison of Clustering Procedures: Octaalanine}
What methods of MSM construction are most robustly able to capture the long-timescale dynamics of protein systems? To address this question, we performed a series of analyses of 27 molecular dynamics trajectories of terminally-blocked octaalanine, a small helix forming peptide. We used 8 different methods to construct the state discretization using clustering with three distance metrics and three clustering algorithms.

For clustering, we considered three distance metrics. The first was the backbone $\phi$ and
$\psi$ dihedral angles. Each conformation was represented by the sine and cosine
of these torsions for a total of 32 features per frame, and distances for
clustering were computed using a Euclidean metric. Second, we considered the distribution of reciprocal interatomic distances (DRID) distance metric introduced by Zhou and
Caflisch,\cite{zhou2012distribution} using the $C\alpha$, $C\beta$, $C$, $N$,
and $O$ atoms in each residue. Finally, we considered the Cartesian minimal
root mean square deviation (RMSD) using the same set of atoms per
residue.\cite{theobald2005rapid} We also considered three clustering algorithms, $k$-centers,\cite{Beauchamp2011Msmbuilder2} a landmark version of UPGMA hierarchical clustering (see \cref{appendix:upgma}), and $k$-means.\cite{Lloyd_1982}

For each pair of distance metric and clustering algorithm (excluding $k$-means \& RMSD which are incompatible),\cite{senne2012emma} we performed 5-fold cross validation using between 10 and 500 states for the clustering. For this experiment, we heuristically chose a lag time of $\tau=10$ ps, and $m=6$, to capture the first five dynamical processes in addition to the stationary distribution. The results are shown in \cref{fig:octaalanine}, with blue curves indicating the mean GMRQ on the training set, and red curves indicating the mean performance on the held-out sets. We find that in all cases, the performance on the training set is \emph{optimistic}, in the sense that the \emph{ansatz} eigenvectors fit during training score more poorly when re{\"e}valuated on held out data. Furthermore, although the training curves all continue to increase with respect to the number of states within the parameter range studied -- which might be interpreted from a variational perspective as the quality of the models continually increasing -- the performance on the test sets tends to peak at a moderate number of states and then decrease. We interpret this as a sign of overfitting: when the number of states is too large, with models fitting the statistical noise in the dataset rather than the underlying slow dynamical processes. Of the parameters studied, the best performance appears to be using the combination of $k$-means clustering with
the dihedral distance metric, using between 50 and 200 states. We also note that $k$-centers appears to yield particularly poor models for all distance metrics, which may be rationalized on the basis that, by design, the algorithm selects outlier conformations to serve as cluster centers.\cite{Beauchamp2011Msmbuilder2}

\begin{figure*}
\centering
\includegraphics[width=468.0pt]{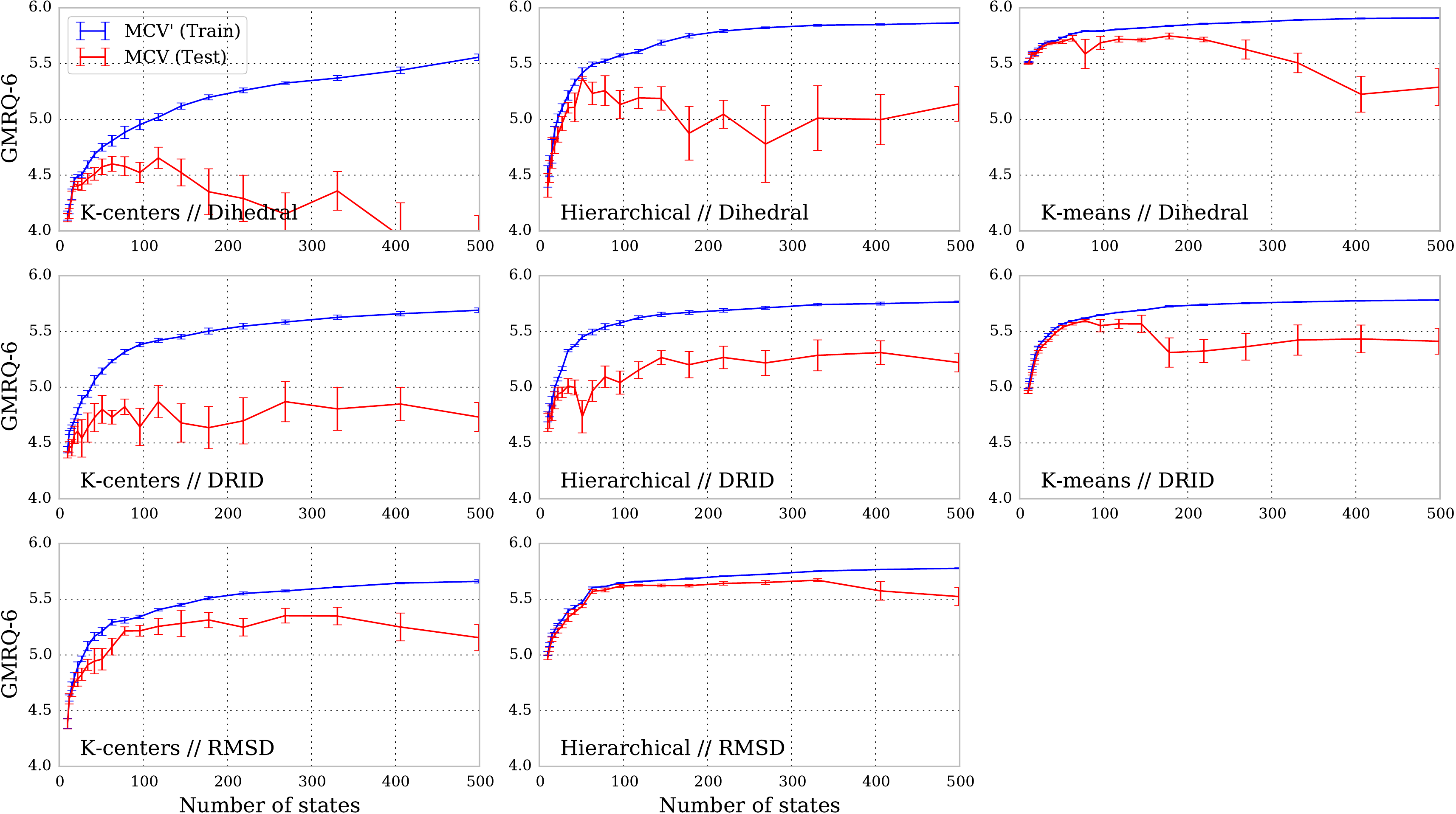}
\caption{\label{fig:octaalanine} Comparison of 8 methods for building MSMs under 5-fold cross validation, evaluated using the rank-6 GMRQ. We used the $k$-centers, $k$-means, and landmark-based ($n_\textrm{landmarks}=5000$) UPGMA hierarchical clustering algorithms, with the DRID\cite{zhou2012distribution} and backbone dihedral angle featurizations. Error bars indicate the standard error in the mean over the cross validation folds.}
\end{figure*}

\section{Discussion}

Some amount of summarization, coarse-graining or dimensionality reduction of
molecular dynamics data sets is a necessary part of their use to answer
questions in biological physics. In this work, we argue that the goal of this
effort should essentially be to find the dominant eigenfunctions of the
system's propagator, an unknown integral operator controlling the system's
dynamics. We show that this goal can be formulated as the variational
optimization of a single scalar functional, which can be approximated using
trajectories obtained from simulation and a parametric basis set. Although
overfitting is a concern with finite simulation data, this risk can be
mitigated by the use of cross-validation.

When the basis sets are restricted to mutually-orthogonal indicator functions
or linear functions of the input coordinates, this method corresponds to the
existing MSM and tICA methods. Unlike previous formulations, it provides a method by which MSM and tICA solutions can be ``scored'' on new data sets that were not used during parameterization, making it possible to measure the generalization performance of these methods and choose the various hyperparameters required for each method, such as the number of MSM states or clustering method. Furthermore, the extension to other families of basis functions (e.g Gaussians) is straightforward, and GMRQ provides a natural quantitative basis on which to conclude whether these new methods are superior to existing basis sets.

\subsection{Connections to quantum mechanics and machine learning}
The variational theorem for eigenspaces in this work has strong connections to work in two other related fields: excited state electronic structure theory in quantum mechanics and multi-class Fisher discriminant analysis in machine learning. In quantum mechanics, \cref{thm:2} is analogous to what has been called the ensemble or trace variational principle in that field,\cite{theophilou79energy, gross88rayleigh, gidopoulos2002ensemble, lai2014density} which bounds the sum of the energy of the first $m$ eigenstates of the Hamiltonian by the trace of a matrix of Ritz values. While the goal of finding just the ground-state eigenfunction ($m=1$) is more common in computational quantum chemistry, the simultaneous optimization of many eigenstates is critical for many applications including band-structure calculations for materials in solid state physics.

Furthermore, in machine learning, this work has an analog in the theory multi-class Fisher discriminant analysis.\cite{rao1948utilization} Here, the goal is to find a low-rank projection of a labeled multi-class dataset which maximizes the between-class variance of the dataset while controlling the within-class variances. The optimal discriminant vectors are shown to be the first $k$ generalized eigenvectors of an eigenproblem involving these two variance matrices -- the problem shares the same structure as \cref{eq:ra} in this work.\cite{baudat2000generalized} We anticipate that this parallel will aid the development of improved algorithms for the identification of slow molecular eigenfunctions, especially with respect to regularization and sparse formulations.\cite{clemmensen2011sparse, sriperumbudur2011majorizarion}

\subsection{Comparison to likelihood maximization}
While we focus on the identification of the dominant eigenfunctions of the system's propagator, a different viewpoint is that analysis of MD should essential entail the construction of probabilistic, generative models over trajectories, fit for example by maximum likelihood or Bayesian methods.

As we show in \cref{sect:msm}, and \citet{nuske2014variational} have shown earlier, MSMs arise naturally from a maximization of \cref{th:varspace} when the \emph{ansatz} eigenfunctions are constrained to be linear combinations of a set of mutually orthogonal indicator functions. However, MSMs can also be viewed directly as probabilistic models, constructed by maximizing a likelihood function of the trajectories with respect to the model parameters. This probabilistic view has, in fact, been central to the field, driving the development of improved methods for example in model selection,\cite{kellogg2012evaluation, McGibbon2014Statistical} parameterization,\cite{rains2010bayesian} and coarse-graining.\cite{bacallado2009bayesian, bowman2012improved} To what extent does this imply that the variational and probabilistic views are equivalent?

In \cref{sect:tension} we show that while these two views may coincide for the particular choice of basis set with MSMs, they need not be equivalent in general. In fact, the GMRQ-optimal model formed by the first $m$ eigenfunctions of the propagator need not be positivity preserving, which is essential to form a probabilistic likelihood function in the sense of \citet{kellogg2012evaluation} or McGibbon, Schwantes and Pande.\cite{McGibbon2014Statistical} On the other hand, the two views \emph{are} tightly coupled; their connection is given by the error bounds proved by \citet{sarich2010approximation}. When the model gives a good approximation to the slow propagator eigenspace (low eigenfunction approximation error, high GMRQ), a good approximation to the long-timescale transition probabilities is obtained.

Cross validation with the log likelihood requires either a generative model for the high dimensional data, such as a hidden Markov model (HMM),\cite{mcgibbon2013understanding} or dimensionality reduction before model comparison. This is a major disadvantage, because accurate and tractable generative models for time series with tens or hundreds of thousands dimensions are not generally available. However, treating dimensionality reduction as a preprocessing and applying probabilistic models afterwards, as done by \citet{McGibbon2014Statistical}, does not enable quantitative comparison between alternative competing dimensionality reduction protocols. With the GMRQ, on the other hand, the need for a reference state decomposition or high-dimensional generative model is eliminated,\cite{bacallado2009bayesian}  and different dimensionality reduction procedures can easily be compared in a quantitative manner, as shown in \cref{fig:octaalanine}.

\section{Conclusions}

The proliferation of new and improved methods for constructing low-dimensional
models of molecular kinetics given a set of high-resolution MD trajectories has
been a boon to the field, but the lack of a unified theoretical framework for
choosing between alternative models has hampered progress, especially for non-experts applying these methods to novel biological systems. In this work we have presented a new variational theorem governing the estimation of the space formed by the span of multiple eigenfunctions of the molecular dynamics propagator. With this method, a single scalar-valued functional scores a proposed model on a supplied data set, and the use of separate testing and training data sets makes it possible to quantify and avoid statistical overfitting. During the training step, time-structure independent components analysis (tICA) and Markov state models (MSMs) are specific instance of this method with different types basis functions. This method extends these tools, making it possible to score trained models on new datasets and perform hyperparameter selection.

We have applied this approach to compare eight different protocols for Markov state model construction on a set of MD simulations of the octaalanine peptide. We find that of the methods tested, $k$-means clustering with the dihedral angles using between 50 and 200 states appears to outperform the other methods, and that the $k$-centers cluster method can be particularly prone to poor generalization performance.
To our knowledge, this work is the first to enable such quantitative and theoretically well-founded comparisons of alternative parameterization strategies for MSMs.

We anticipate that this work will open the door to more complete automation and optimization of MSM construction. Free and open source software fully implementing these methods is available in the MDTraj and MSMBuilder3 packages from \url{http://mdtraj.org} and \url{http://msmbuilder.org}, along with example and tutorials.\cite{McGibbon2014MDTraj} While the lag time, $\tau$ and rank, $m$, of the desired model must be manually specified, other key hyperparameters that control difficult-to-judge statistical tradeoffs, such as the number of states in an MSM, can be chosen be optimizing the cross-validation performance. Furthermore, given recent advances in automated hyperparameter optimization in machine learning, we anticipate that this search itself can be fully automated.\cite{NIPS2012_4522}

\section*{Acknowledgements}
This work was supported by the National Science Foundation and National Institutes of Health under Nos. NIH R01-GM62868, NIH S10 SIG 1S10RR02664701, and NSF MCB-0954714. We thank the anonymous reviewers for their many suggestions for improving this work, Christian R. Schwantes, Mohammad M. Sultan,  Sergio Bacallado, and Krikamol Muandet for helpful discussions, and Jeffrey K. Weber for access to the octaalanine simulations.

\appendix

\section{Proofs of \cref{thm:2} and \cref{lemma}}
\label{appendix:proof}
\subsection*{Proof of \cref{thm:2}}
The eigenfunctions, $\phi_i$, of $\mathcal{P}(\tau)$ form a complete basis. Expand each $f_i = \sum_a w_{ai} \phi_a$ with coefficients $W \in \mathbb{R}^{\infty \times m}$ with $W_{ni} = \langle f_i, \phi_n \rangle_{\mu^{-1}}$.

\begin{align}
P_{ij} &= \langle f_i,\; \mathcal{P} \circ f_j \rangle_{\mu^{-1}} \\
&= \Big\langle \sum_a W_{ai} \phi_a, \; \mathcal{P} \circ \sum_b W_{bj} \phi_b \Big\rangle_{\mu^{-1}} \\
&= \sum_{a} W_{ai} W_{aj} \lambda_a \label{eq:d3}\\
Q_{ij} &= \langle f_i, f_j \rangle_{\mu^{-1}} \\
& = \Big\langle  \sum_a W_{ai} \phi_a, \; \sum_b W_{bj} \phi_b \Big\rangle_{\mu^{-1}} \\
&= \sum_a W_{ai} W_{aj}  \label{eq:d5}
\end{align}

We define the diagonal matrix $D(\lambda)$ with $D_{ii}=\lambda_i$.
Then, \cref{eq:d3} and \cref{eq:d5} can be rewritten in matrix form:
\begin{align}
  P &= W^T D(\lambda) W \label{eq:d7},\\
  Q &= W^T W.  \label{eq:d8}
\end{align}

Let $F=Q^{1/2} \in \mathbb{S}^m_{++}$ be the (unique) positive definite square root of $Q$, which is guaranteed to exist because $Q$ is positive definite, and $B = WF^{-1}$. Then, rearrange the objective function using the cyclic property of the trace:
\begin{align}
\mathcal{R}_\mathcal{P}[f] &= \Tr\big( \underbrace{W^T D(\lambda) W}_P \underbrace{(F F)^{-1}}_{Q^{-1}} \big), \\
&= \Tr\big( F^{-1} W^T  D(\lambda) W F^{-1} \big), \\
&= \Tr\big( B^T  D(\lambda) B \big).
\end{align}
Note that $B^T B = F^{-1} W^T W F^{-1} = I_m$. Therefore, by application of the Ky Fan theorem,\cite{fan1949theorem, overton1992sum}
\begin{equation}
\mathcal{R}_\mathcal{P}[f] \leq \sum_i^m \lambda_i,
\end{equation}
and the equality holds when $f = \{\phi_1, \phi_2, \ldots, \phi_m\}$.

\subsection*{Proof of \cref{lemma}}
Following \citet{absil2002grassmann}, let $f = \{f_1, f_2, \ldots, f_m\}$, and $M \in \mathbb{R}^{m \times m}$ be an arbitrary invertible matrix. Define a new collection of functions $g = \{g_1, g_2, \ldots, g_m\}$, such that $g_j = \sum_{i=1}^m M_{ij} f_i$, and a matrix $W' \in R^{\infty \times m}$ such that $W'_{ni} = \langle g_i, \phi_n \rangle_{\mu^{-1}}$. Expanding the matrix elements of $W'$, we note that
\begin{align}
  W' = WM.
\end{align}
Then, using \cref{eq:d7} and \cref{eq:d8}, $\mathcal{R}_\mathcal{P}[g]$ can be written as a matrix expression involving $W'$, and subsequently rewritten involving $W$ and $M$. Expansion of the matrix products and application of the cyclic property of the trace confirms that $\mathcal{R}_\mathcal{P}[g] = \mathcal{R}_\mathcal{P}[f]$:
\begin{align}
  \mathcal{R}_\mathcal{P}[g] &=  \Tr \big( W'^T D(\lambda) W' (W'^T W')^{-1}), \\
  &= \Tr \big( (WM)^T D(\lambda) (WM) ((WM)^T (WM))^{-1} \big), \\
  &= \Tr \big( M^T W^T D(\lambda) W M^{-1} (W^T W)^{-1} M^{-T} \big), \\
  &= \Tr \big( W^T D(\lambda) W (W^T W)^{-1} \big), \\
  &= \mathcal{R}_\mathcal{P}[f].
\end{align}

The significance of this result is that it demonstrates $\mathcal{R}_\mathcal{P}$ to be invariant to linear transformations of $f$ which preserve the space spanned by the functions. Much like the Ritz value of an trial vector is invariant to rescaling, or the angle between two planes is invariant to linear transformations of the basis vectors defining the planes,  $\mathcal{R}_\mathcal{P}[f]$ is only a functional of the space spanned by $f$. This space -- the set of all $m$-dimensional linear subspaces of a vector or Hilbert space -- is referred to as a \emph{Grassmann manifold}.\cite{absil2002grassmann}

\section{Tension Between Spectral and Probabilistic Approaches}
\label{sect:tension}
Here we show, by way of a simple analytical example, the extent to which the variational and probabilistic approaches to the analysis of molecular dynamics data are indeed distinct. By explicitly constructing the propagator eigenfunctions for a Brownian harmonic oscillator, we show that the rank-$m$ truncated propagator, $\mathcal{P}_m(\tau)$, built from the first $m$ eigenpairs of $\mathcal{P}(\tau)$ is not in general a nonnegativity-preserving operator. That is, for some valid initial distributions, $p_t(\mathbf{x})$, the propagated distribution, $\tilde{p}_{t+\tau}^{(m)}(\mathbf{x}) = \mathcal{P}_m(\tau) \circ p_t(\mathbf{x})$, fails to be non-negative throughout $\Omega$ and thus does not represent a valid probability distribution.
\begin{align}
\label{eq:generator}
p_{t+\tau}^{(m)}(\mathbf{x}) &\ngeq 0\; \forall\; \mathbf{x} \in \Omega
\end{align}
This indicates that variational and probabilistic approaches have the potential to be almost contradictory in what they judge to be ``good'' models of molecular kinetics.

Consider the diffusion of a Brownian particle in the potential $U(x) = x^2$. For simplicity, we take the temperature and diffusion constant to be unity. This is an Ornstein-Uhlenbeck process, and the dynamics are described by the Smoluchowski equation,
\begin{align}
\frac{\partial}{\partial t} p_t(x) &= \mathcal{L} \circ p_t(x),
\end{align}
with infinitesimal generator $\mathcal{L}$ given by
\begin{align}
\mathcal{L} &= \frac{\partial^2}{\partial x^2} + 2 \frac{\partial}{\partial x}x,
\end{align}
and stationary distribution $\mu(x) = \pi^{-1/2} e^{-x^2}$.

We can expand the generator in terms of its eigenfunctions, $\phi_n(x)$, and eigenvalues, $\xi_n$, defined by,
\begin{align}
\mathcal{L} \circ \phi_n(x) = \xi_n \phi_n(x),
\end{align}
which can be recognized as the Hermite equation whose solutions are related to the Hermite polynomials, $H_n$. For $n=\{0, 1, \ldots\}$ the solutions are
\begin{align}
\phi_n(x) &= c_n e^{-x^2} H_n(x), \\
\xi_n &= -2n, \\
c_n^{2} &= (2^n n! \pi)^{-1},
\end{align}
where the normalizing constants, $c_n$, are chosen such that $\langle \phi_n, \phi_m \rangle_{\mu^{-1}} = \delta_{nm}$.

The propagator $\mathcal{P}(\tau)$ can be formed by integrating \cref{eq:generator} with respect to $t$, giving
\begin{equation}
\mathcal{P}(\tau) = e^{\tau \mathcal{L}}.
\end{equation}
$\mathcal{P}(\tau)$ shares the same eigenfunctions as $\mathcal{L}$. Its eigenvalues, $\lambda_n$, are related to the eigenvalues of $\mathcal{L}$ by
\begin{align}
    \lambda_n = e^{-\tau \xi_n}.
\end{align}

We now define the rank-$m$ truncated propagator, $\mathcal{P}_m(\tau)$, such that
\begin{align}
&\mathcal{P}_m(\tau) \circ p_t = \sum_{n=0}^{m-1} \lambda_n \langle p_t, \phi_n \rangle_{\mu^{-1}} \phi_n \\
\label{eq:11}
&= \sum_{n=0}^{m-1} e^{-2n\tau} c_n e^{-x^2} H_n(x) \left[\int_{-\infty}^{\infty} dx' \; c_n \sqrt{\pi} \, p_t(x') H_n(x') \right]
\end{align}
Consider an initial distribution, $p_t(x) = \delta(x-x_0)$, propagated forward in time by $\mathcal{P}_m$. Let $\tilde{p}_{\tau}^{(m)} = \mathcal{P}_m(\tau) \circ \delta(x-x_0)$. Then, \cref{eq:11} simplifies to
\begin{align}
\tilde{p}_{\tau}^{(m)}(x) = \sum_{n=0}^{m-1} \frac{1}{2^n n! \sqrt{\pi}} e^{-2n\tau} \, e^{-x^2} H_n(x) H_n(x_0).
\end{align}

Consider now the specific case of $m=2$. Using the explicit expansion $H_0(x) = 1$, and $H_1(x) = 2x$, we have
\begin{align}
\label{eq:m2}
\tilde{p}_{\tau}^{(2)}(x) = \frac{1}{\sqrt{\pi}}e^{-x^2} \left(1 + 2 x x_0 e^{-2\tau} \right).
\end{align}

Note that \cref{eq:m2} has a zero when $x = -e^{2\tau}/2x_0$, and that
\begin{align}
\tilde{p}_{\tau}^{(2)}(x) < 0\;\; \mbox{when } \begin{cases}
x < -e^{2\tau}/2x_0 &\mbox{if } x_0 > 0 \\
x > -e^{2\tau}/2x_0 &\mbox{if } x_0 < 0.
\end{cases}
\end{align}
Because of this non-positivity, $\tilde{p}_{\tau}^{(2)}(x)$ is not a valid probability distribution.

This example demonstrates that the rank-$m$ truncated propagator need not, in general, preserve the positivity of distributions it acts on. Therefore, if such a model of the dynamics are fit or assessed via maximum-likelihood methods on datasets consisting of observed transitions, despite being \emph{optimal} by spectral norm, the true rank-$m$ truncated propagator may appear to give a log likelihood of $-\infty$. The variational and probabilistic approaches to modeling molecular kinetics can indeed be very different.

\section{Double-well Potential Integrator and Eigenfunctions}
\label{appendix:analytic}

To discretize the Brownian dynamics stochastic differential equation in \cref{eq:sde} with reflecting boundary conditions at $-\pi$ and $\pi$, we used the Euler integrator,
\begin{align}
  \label{eq:euler}
  x_{t+1} = bc\left(x_t + \left(\nabla V(x_t) + \sqrt{2D} R(t)\right)\Delta t\right),
\end{align}
where steps that went outside the boundaries by a given distance were reflected back into the interval with a matching distance to the boundary:
\begin{align}
  bc(x) = \begin{cases}
    2\pi-x &\text{if $x > \pi$},\\
    -2\pi-x &\text{if $x < -\pi$},\\
    x &\text{otherwise}.
  \end{cases}
\end{align}

We computed the propagator eigenvalues by discretizing the interval into $n$ MSM states $\{s_i\}_{i=1}^n$, following \cref{eq:si}, and computing the matrix elements without stochastic sampling. This calculation is more straightforward by working with the the transition matrix $T \in \mathbb{R}^{n \times n}$:
\begin{align}
T_{ij} = \mathbb{P}\left[x_{t+\tau} \in s_j | x_t \in s_i\right],
\end{align}
Instead of the correlation and overlap matrices, $C$ and $S$, directly.
Note that as shown by \citet{nuske2014variational} and \citet{prinz2011markov}, $T=S^{-1}C$. Thus the eigenvalues of $T$ are identical to the generalized eigenvalues of $(C,S)$.

To calculate the matrix elements of $T$, we consider each each state, $s_i$, represented by its left endpoint, $x_i$. For each pair of states $(i,j)$, we calculate the probability of the random force required to transition between them in one step, from \cref{eq:euler}, taking into account the fact that because of the reflecting boundary conditions, the transition could have taken place via a transition outside the interval followed by a reflection.

Let $\delta x$ be the width of the states, $\delta x = n^{-1}2\pi$. For each $i \in \{1, \ldots n\}$ and $k \in \{-n, \ldots, n\}$, we calculate the action for
a step from $x_i$ to $x_i + k\delta x$.
\begin{align}
a_{ik} = \frac{1}{\sqrt{2\pi}}\exp\left(\frac{((x_i + k\delta x) - x_i + \nabla V(x_i)\Delta t)}{\sqrt{2D}}\right).
\end{align}

Let $j'_{ik} \in \{1, \ldots, n\}$ be the index of state which contains $bc(x_i + k\delta x)$. We calculate  $T_{ij}$ by summing the appropriate action terms which, because of the boundary conditions, get reflected into the same ending state:
\begin{align}
  T_{ij} = \sum_{k = -n}^{n} a_{ik} \; \delta_{j'_{ik},j}.
\end{align}
where $\delta_{i,j}$ is the Kronecker delta. This calculation is implemented in the file \texttt{brownian1d.py} in the MSMBuilder3 package. The eigenvalues of $T$ converge rapidly as $n$ increases. Our results in \cref{fig:doublewell} use $n=500$.

\section{Landmark UPGMA Clustering}
\label{appendix:upgma}
Landmark-based UPGMA (Unweighted Pair Group Method with Arithmetic Mean) agglomerative clustering is a simple scalable hierarchical clustering which does not require computing the full matrix of pairwise distances between all data points. The procedure first subsamples $l$ ``landmark'' data points at regular intervals from the input data. These data points are then clustered using the standard algorithm, resulting in $n$ clusters. \cite{Mullner2013fastcluster} Let $S_n$ be the set of landmark data points assigned by the algorithm to the cluster $n$, and $d(x, x')$ be the distance
metric employed. Then, each remaining data point in the training set as well
as new data points from the test set, $x^*$, are assigned to cluster, $s(x^*) \in \{1,\ldots,n\}$, whose landmarks they are on average closest to:
\begin{align}
s(x^*) = \argmin_{n} \frac{1}{|S_n|} \sum_{x \in S_n} d(x^*, x).
\end{align}

\section{Simulation Setup}
\label{appendix:octaalanine}
We performed all-atom molecular dynamics simulations of terminally-blocked octaalanine (Ace-(Ala)$_8$-NHMe) in explicit solvent using the GROMACS 4 simulation package,\cite{hess2008gromacs} the AMBER ff99SB-ILDN-NMR forcefield,\cite{liNMR2010} and the TIP3P water model.\cite{jorgensen1983comparison} The system was energy minimized, followed by 1 ns of equilibration using the velocity rescaling thermostat (reference temperature of 298K, time constant of 0.1 ps),\cite{bussi2007canonical} Parrinello-Rahman barostat (reference pressure of 1 bar, time constant of 1 ps, isotropic compressibility of $5 \times 10^{-5}$ bar),\cite{parrinello1981polymorphic} and Verlet integrator (time step of 2 fs). Production simulations were performed in the canonical ensemble using the same integrator and thermostat. Nonbonded interactions in all cases were treated with the particle mesh Ewald method, using a real space cutoff distance for Ewald summation as well as for van der Waals interactions of 10.0 $\mathrm{\AA}$.\cite{darden1993particle} Twenty six such simulations were performed, with production lengths between 20 and 150 ns each. The total aggregate sampling was 1.74 $\mu s$.

\bibliography{bibliography}

\begin{thebibliography}{87}%
\makeatletter
\providecommand \@ifxundefined [1]{%
 \@ifx{#1\undefined}
}%
\providecommand \@ifnum [1]{%
 \ifnum #1\expandafter \@firstoftwo
 \else \expandafter \@secondoftwo
 \fi
}%
\providecommand \@ifx [1]{%
 \ifx #1\expandafter \@firstoftwo
 \else \expandafter \@secondoftwo
 \fi
}%
\providecommand \natexlab [1]{#1}%
\providecommand \enquote  [1]{``#1''}%
\providecommand \bibnamefont  [1]{#1}%
\providecommand \bibfnamefont [1]{#1}%
\providecommand \citenamefont [1]{#1}%
\providecommand \href@noop [0]{\@secondoftwo}%
\providecommand \href [0]{\begingroup \@sanitize@url \@href}%
\providecommand \@href[1]{\@@startlink{#1}\@@href}%
\providecommand \@@href[1]{\endgroup#1\@@endlink}%
\providecommand \@sanitize@url [0]{\catcode `\\12\catcode `\$12\catcode
  `\&12\catcode `\#12\catcode `\^12\catcode `\_12\catcode `\%12\relax}%
\providecommand \@@startlink[1]{}%
\providecommand \@@endlink[0]{}%
\providecommand \url  [0]{\begingroup\@sanitize@url \@url }%
\providecommand \@url [1]{\endgroup\@href {#1}{\urlprefix }}%
\providecommand \urlprefix  [0]{URL }%
\providecommand \Eprint [0]{\href }%
\providecommand \doibase [0]{http://dx.doi.org/}%
\providecommand \selectlanguage [0]{\@gobble}%
\providecommand \bibinfo  [0]{\@secondoftwo}%
\providecommand \bibfield  [0]{\@secondoftwo}%
\providecommand \translation [1]{[#1]}%
\providecommand \BibitemOpen [0]{}%
\providecommand \bibitemStop [0]{}%
\providecommand \bibitemNoStop [0]{.\EOS\space}%
\providecommand \EOS [0]{\spacefactor3000\relax}%
\providecommand \BibitemShut  [1]{\csname bibitem#1\endcsname}%
\let\auto@bib@innerbib\@empty
\bibitem [{\citenamefont {Dill}\ \emph {et~al.}(2008)\citenamefont {Dill},
  \citenamefont {Ozkan}, \citenamefont {Shell},\ and\ \citenamefont
  {Weikl}}]{dill2008protein}%
  \BibitemOpen
  \bibfield  {author} {\bibinfo {author} {\bibfnamefont {K.~A.}\ \bibnamefont
  {Dill}}, \bibinfo {author} {\bibfnamefont {S.~B.}\ \bibnamefont {Ozkan}},
  \bibinfo {author} {\bibfnamefont {M.~S.}\ \bibnamefont {Shell}}, \ and\
  \bibinfo {author} {\bibfnamefont {T.~R.}\ \bibnamefont {Weikl}},\ }\href
  {\doibase 10.1146/annurev.biophys.37.092707.153558} {\bibfield  {journal}
  {\bibinfo  {journal} {Annu. Rev. Biophys.}\ }\textbf {\bibinfo {volume}
  {37}},\ \bibinfo {pages} {289} (\bibinfo {year} {2008})}\BibitemShut
  {NoStop}%
\bibitem [{\citenamefont {Huse}\ and\ \citenamefont
  {Kuriyan}(2002)}]{huse2002conformational}%
  \BibitemOpen
  \bibfield  {author} {\bibinfo {author} {\bibfnamefont {M.}~\bibnamefont
  {Huse}}\ and\ \bibinfo {author} {\bibfnamefont {J.}~\bibnamefont {Kuriyan}},\
  }\href {\doibase 10.1016/S0092-8674(02)00741-9} {\bibfield  {journal}
  {\bibinfo  {journal} {Cell}\ }\textbf {\bibinfo {volume} {109}},\ \bibinfo
  {pages} {275} (\bibinfo {year} {2002})}\BibitemShut {NoStop}%
\bibitem [{\citenamefont {Vargas}\ \emph {et~al.}(2012)\citenamefont {Vargas},
  \citenamefont {Yarov-Yarovoy}, \citenamefont {Khalili-Araghi}, \citenamefont
  {Catterall}, \citenamefont {Klein}, \citenamefont {Tarek}, \citenamefont
  {Lindahl}, \citenamefont {Schulten}, \citenamefont {Perozo}, \citenamefont
  {Bezanilla} \emph {et~al.}}]{vargas2012emerging}%
  \BibitemOpen
  \bibfield  {author} {\bibinfo {author} {\bibfnamefont {E.}~\bibnamefont
  {Vargas}}, \bibinfo {author} {\bibfnamefont {V.}~\bibnamefont
  {Yarov-Yarovoy}}, \bibinfo {author} {\bibfnamefont {F.}~\bibnamefont
  {Khalili-Araghi}}, \bibinfo {author} {\bibfnamefont {W.~A.}\ \bibnamefont
  {Catterall}}, \bibinfo {author} {\bibfnamefont {M.~L.}\ \bibnamefont
  {Klein}}, \bibinfo {author} {\bibfnamefont {M.}~\bibnamefont {Tarek}},
  \bibinfo {author} {\bibfnamefont {E.}~\bibnamefont {Lindahl}}, \bibinfo
  {author} {\bibfnamefont {K.}~\bibnamefont {Schulten}}, \bibinfo {author}
  {\bibfnamefont {E.}~\bibnamefont {Perozo}}, \bibinfo {author} {\bibfnamefont
  {F.}~\bibnamefont {Bezanilla}},  \emph {et~al.},\ }\href {\doibase
  10.1085/jgp.201210873} {\bibfield  {journal} {\bibinfo  {journal} {J. Gen.
  Physiol.}\ }\textbf {\bibinfo {volume} {140}},\ \bibinfo {pages} {587}
  (\bibinfo {year} {2012})}\BibitemShut {NoStop}%
\bibitem [{\citenamefont {Phillips}\ \emph {et~al.}(2013)\citenamefont
  {Phillips}, \citenamefont {Zhang}, \citenamefont {Cunningham}, \citenamefont
  {Zhou}, \citenamefont {Forrest}, \citenamefont {Liu}, \citenamefont
  {Steffek}, \citenamefont {Lee}, \citenamefont {Tam}, \citenamefont {Helgason}
  \emph {et~al.}}]{phillips2013conformational}%
  \BibitemOpen
  \bibfield  {author} {\bibinfo {author} {\bibfnamefont {A.~H.}\ \bibnamefont
  {Phillips}}, \bibinfo {author} {\bibfnamefont {Y.}~\bibnamefont {Zhang}},
  \bibinfo {author} {\bibfnamefont {C.~N.}\ \bibnamefont {Cunningham}},
  \bibinfo {author} {\bibfnamefont {L.}~\bibnamefont {Zhou}}, \bibinfo {author}
  {\bibfnamefont {W.~F.}\ \bibnamefont {Forrest}}, \bibinfo {author}
  {\bibfnamefont {P.~S.}\ \bibnamefont {Liu}}, \bibinfo {author} {\bibfnamefont
  {M.}~\bibnamefont {Steffek}}, \bibinfo {author} {\bibfnamefont
  {J.}~\bibnamefont {Lee}}, \bibinfo {author} {\bibfnamefont {C.}~\bibnamefont
  {Tam}}, \bibinfo {author} {\bibfnamefont {E.}~\bibnamefont {Helgason}},
  \emph {et~al.},\ }\href {\doibase 10.1073/pnas.1302407110} {\bibfield
  {journal} {\bibinfo  {journal} {Proc. Natl. Acad. Sci. U.S.A.}\ }\textbf
  {\bibinfo {volume} {110}},\ \bibinfo {pages} {11379} (\bibinfo {year}
  {2013})}\BibitemShut {NoStop}%
\bibitem [{\citenamefont {Careri}\ \emph {et~al.}(1975)\citenamefont {Careri},
  \citenamefont {Fasella}, \citenamefont {Gratton},\ and\ \citenamefont
  {Jencks}}]{careri1975statistical}%
  \BibitemOpen
  \bibfield  {author} {\bibinfo {author} {\bibfnamefont {G.}~\bibnamefont
  {Careri}}, \bibinfo {author} {\bibfnamefont {P.}~\bibnamefont {Fasella}},
  \bibinfo {author} {\bibfnamefont {E.}~\bibnamefont {Gratton}}, \ and\
  \bibinfo {author} {\bibfnamefont {W.}~\bibnamefont {Jencks}},\ }\href
  {\doibase 10.3109/10409237509102555} {\bibfield  {journal} {\bibinfo
  {journal} {Crit. Rev. Biochem. Mol. Biol.}\ }\textbf {\bibinfo {volume}
  {3}},\ \bibinfo {pages} {141} (\bibinfo {year} {1975})}\BibitemShut {NoStop}%
\bibitem [{\citenamefont {Buergi}\ and\ \citenamefont
  {Dunitz}(1983)}]{buergi1983crystal}%
  \BibitemOpen
  \bibfield  {author} {\bibinfo {author} {\bibfnamefont {H.~B.}\ \bibnamefont
  {Buergi}}\ and\ \bibinfo {author} {\bibfnamefont {J.~D.}\ \bibnamefont
  {Dunitz}},\ }\href {\doibase 10.1021/ar00089a002} {\bibfield  {journal}
  {\bibinfo  {journal} {Acc. Chem. Res.}\ }\textbf {\bibinfo {volume} {16}},\
  \bibinfo {pages} {153} (\bibinfo {year} {1983})}\BibitemShut {NoStop}%
\bibitem [{\citenamefont {Mittermaier}\ and\ \citenamefont
  {Kay}(2006)}]{mittermaier2006new}%
  \BibitemOpen
  \bibfield  {author} {\bibinfo {author} {\bibfnamefont {A.}~\bibnamefont
  {Mittermaier}}\ and\ \bibinfo {author} {\bibfnamefont {L.~E.}\ \bibnamefont
  {Kay}},\ }\href {\doibase 10.1126/science.1124964} {\bibfield  {journal}
  {\bibinfo  {journal} {Science}\ }\textbf {\bibinfo {volume} {312}},\ \bibinfo
  {pages} {224} (\bibinfo {year} {2006})}\BibitemShut {NoStop}%
\bibitem [{\citenamefont {Wang}, \citenamefont {Chen},\ and\ \citenamefont
  {Van~Voorhis}(2012)}]{wang2012systematic}%
  \BibitemOpen
  \bibfield  {author} {\bibinfo {author} {\bibfnamefont {L.-P.}\ \bibnamefont
  {Wang}}, \bibinfo {author} {\bibfnamefont {J.}~\bibnamefont {Chen}}, \ and\
  \bibinfo {author} {\bibfnamefont {T.}~\bibnamefont {Van~Voorhis}},\ }\href
  {\doibase 10.1021/ct300826t} {\bibfield  {journal} {\bibinfo  {journal} {J.
  Chem. Theory Comput.}\ }\textbf {\bibinfo {volume} {9}},\ \bibinfo {pages}
  {452} (\bibinfo {year} {2012})}\BibitemShut {NoStop}%
\bibitem [{\citenamefont {Huang}\ and\ \citenamefont
  {Roux}(2013)}]{huang2013automated}%
  \BibitemOpen
  \bibfield  {author} {\bibinfo {author} {\bibfnamefont {L.}~\bibnamefont
  {Huang}}\ and\ \bibinfo {author} {\bibfnamefont {B.}~\bibnamefont {Roux}},\
  }\href {\doibase 10.1021/ct4003477} {\bibfield  {journal} {\bibinfo
  {journal} {J. Chem. Theory Comput.}\ }\textbf {\bibinfo {volume} {9}},\
  \bibinfo {pages} {3543} (\bibinfo {year} {2013})}\BibitemShut {NoStop}%
\bibitem [{\citenamefont {Ponder}\ \emph {et~al.}(2010)\citenamefont {Ponder},
  \citenamefont {Wu}, \citenamefont {Ren}, \citenamefont {Pande}, \citenamefont
  {Chodera}, \citenamefont {Schnieders}, \citenamefont {Haque}, \citenamefont
  {Mobley}, \citenamefont {Lambrecht}, \citenamefont {DiStasio~Jr} \emph
  {et~al.}}]{ponder2010current}%
  \BibitemOpen
  \bibfield  {author} {\bibinfo {author} {\bibfnamefont {J.~W.}\ \bibnamefont
  {Ponder}}, \bibinfo {author} {\bibfnamefont {C.}~\bibnamefont {Wu}}, \bibinfo
  {author} {\bibfnamefont {P.}~\bibnamefont {Ren}}, \bibinfo {author}
  {\bibfnamefont {V.~S.}\ \bibnamefont {Pande}}, \bibinfo {author}
  {\bibfnamefont {J.~D.}\ \bibnamefont {Chodera}}, \bibinfo {author}
  {\bibfnamefont {M.~J.}\ \bibnamefont {Schnieders}}, \bibinfo {author}
  {\bibfnamefont {I.}~\bibnamefont {Haque}}, \bibinfo {author} {\bibfnamefont
  {D.~L.}\ \bibnamefont {Mobley}}, \bibinfo {author} {\bibfnamefont {D.~S.}\
  \bibnamefont {Lambrecht}}, \bibinfo {author} {\bibfnamefont {R.~A.}\
  \bibnamefont {DiStasio~Jr}},  \emph {et~al.},\ }\href {\doibase
  10.1021/jp910674d} {\bibfield  {journal} {\bibinfo  {journal} {J. Phys. Chem.
  B}\ }\textbf {\bibinfo {volume} {114}},\ \bibinfo {pages} {2549} (\bibinfo
  {year} {2010})}\BibitemShut {NoStop}%
\bibitem [{\citenamefont {Best}\ \emph {et~al.}(2012)\citenamefont {Best},
  \citenamefont {Zhu}, \citenamefont {Shim}, \citenamefont {Lopes},
  \citenamefont {Mittal}, \citenamefont {Feig},\ and\ \citenamefont
  {MacKerell~Jr}}]{best2012optimization}%
  \BibitemOpen
  \bibfield  {author} {\bibinfo {author} {\bibfnamefont {R.~B.}\ \bibnamefont
  {Best}}, \bibinfo {author} {\bibfnamefont {X.}~\bibnamefont {Zhu}}, \bibinfo
  {author} {\bibfnamefont {J.}~\bibnamefont {Shim}}, \bibinfo {author}
  {\bibfnamefont {P.~E.}\ \bibnamefont {Lopes}}, \bibinfo {author}
  {\bibfnamefont {J.}~\bibnamefont {Mittal}}, \bibinfo {author} {\bibfnamefont
  {M.}~\bibnamefont {Feig}}, \ and\ \bibinfo {author} {\bibfnamefont {A.~D.}\
  \bibnamefont {MacKerell~Jr}},\ }\href {\doibase 10.1021/ct300400x} {\bibfield
   {journal} {\bibinfo  {journal} {J. Chem. Theory Comput.}\ }\textbf {\bibinfo
  {volume} {8}},\ \bibinfo {pages} {3257} (\bibinfo {year} {2012})}\BibitemShut
  {NoStop}%
\bibitem [{\citenamefont {Lopes}\ \emph {et~al.}(2013)\citenamefont {Lopes},
  \citenamefont {Huang}, \citenamefont {Shim}, \citenamefont {Luo},
  \citenamefont {Li}, \citenamefont {Roux},\ and\ \citenamefont
  {MacKerell~Jr}}]{lopes2013polarizable}%
  \BibitemOpen
  \bibfield  {author} {\bibinfo {author} {\bibfnamefont {P.~E.}\ \bibnamefont
  {Lopes}}, \bibinfo {author} {\bibfnamefont {J.}~\bibnamefont {Huang}},
  \bibinfo {author} {\bibfnamefont {J.}~\bibnamefont {Shim}}, \bibinfo {author}
  {\bibfnamefont {Y.}~\bibnamefont {Luo}}, \bibinfo {author} {\bibfnamefont
  {H.}~\bibnamefont {Li}}, \bibinfo {author} {\bibfnamefont {B.}~\bibnamefont
  {Roux}}, \ and\ \bibinfo {author} {\bibfnamefont {A.~D.}\ \bibnamefont
  {MacKerell~Jr}},\ }\href {\doibase 10.1021/ct400781b} {\bibfield  {journal}
  {\bibinfo  {journal} {J. Chem. Theory Comput.}\ }\textbf {\bibinfo {volume}
  {9}},\ \bibinfo {pages} {5430} (\bibinfo {year} {2013})}\BibitemShut
  {NoStop}%
\bibitem [{\citenamefont {Stone}\ \emph {et~al.}(2010)\citenamefont {Stone},
  \citenamefont {Hardy}, \citenamefont {Ufimtsev},\ and\ \citenamefont
  {Schulten}}]{stone2010gpu}%
  \BibitemOpen
  \bibfield  {author} {\bibinfo {author} {\bibfnamefont {J.~E.}\ \bibnamefont
  {Stone}}, \bibinfo {author} {\bibfnamefont {D.~J.}\ \bibnamefont {Hardy}},
  \bibinfo {author} {\bibfnamefont {I.~S.}\ \bibnamefont {Ufimtsev}}, \ and\
  \bibinfo {author} {\bibfnamefont {K.}~\bibnamefont {Schulten}},\ }\href
  {\doibase 10.1016/j.jmgm.2010.06.010} {\bibfield  {journal} {\bibinfo
  {journal} {J. Mol. Graphics Modell.}\ }\textbf {\bibinfo {volume} {29}},\
  \bibinfo {pages} {116} (\bibinfo {year} {2010})}\BibitemShut {NoStop}%
\bibitem [{\citenamefont {Shaw}\ \emph {et~al.}(2009)\citenamefont {Shaw},
  \citenamefont {Dror}, \citenamefont {Salmon}, \citenamefont {Grossman},
  \citenamefont {Mackenzie}, \citenamefont {Bank}, \citenamefont {Young},
  \citenamefont {Deneroff}, \citenamefont {Batson}, \citenamefont {Bowers},
  \citenamefont {Chow}, \citenamefont {Eastwood}, \citenamefont {Ierardi},
  \citenamefont {Klepeis}, \citenamefont {Kuskin}, \citenamefont {Larson},
  \citenamefont {Lindorff-Larsen}, \citenamefont {Maragakis}, \citenamefont
  {Moraes}, \citenamefont {Piana}, \citenamefont {Shan},\ and\ \citenamefont
  {Towles}}]{shaw2008anton}%
  \BibitemOpen
  \bibfield  {author} {\bibinfo {author} {\bibfnamefont {D.~E.}\ \bibnamefont
  {Shaw}}, \bibinfo {author} {\bibfnamefont {R.~O.}\ \bibnamefont {Dror}},
  \bibinfo {author} {\bibfnamefont {J.~K.}\ \bibnamefont {Salmon}}, \bibinfo
  {author} {\bibfnamefont {J.~P.}\ \bibnamefont {Grossman}}, \bibinfo {author}
  {\bibfnamefont {K.~M.}\ \bibnamefont {Mackenzie}}, \bibinfo {author}
  {\bibfnamefont {J.~A.}\ \bibnamefont {Bank}}, \bibinfo {author}
  {\bibfnamefont {C.}~\bibnamefont {Young}}, \bibinfo {author} {\bibfnamefont
  {M.~M.}\ \bibnamefont {Deneroff}}, \bibinfo {author} {\bibfnamefont
  {B.}~\bibnamefont {Batson}}, \bibinfo {author} {\bibfnamefont {K.~J.}\
  \bibnamefont {Bowers}}, \bibinfo {author} {\bibfnamefont {E.}~\bibnamefont
  {Chow}}, \bibinfo {author} {\bibfnamefont {M.~P.}\ \bibnamefont {Eastwood}},
  \bibinfo {author} {\bibfnamefont {D.~J.}\ \bibnamefont {Ierardi}}, \bibinfo
  {author} {\bibfnamefont {J.~L.}\ \bibnamefont {Klepeis}}, \bibinfo {author}
  {\bibfnamefont {J.~S.}\ \bibnamefont {Kuskin}}, \bibinfo {author}
  {\bibfnamefont {R.~H.}\ \bibnamefont {Larson}}, \bibinfo {author}
  {\bibfnamefont {K.}~\bibnamefont {Lindorff-Larsen}}, \bibinfo {author}
  {\bibfnamefont {P.}~\bibnamefont {Maragakis}}, \bibinfo {author}
  {\bibfnamefont {M.~A.}\ \bibnamefont {Moraes}}, \bibinfo {author}
  {\bibfnamefont {S.}~\bibnamefont {Piana}}, \bibinfo {author} {\bibfnamefont
  {Y.}~\bibnamefont {Shan}}, \ and\ \bibinfo {author} {\bibfnamefont
  {B.}~\bibnamefont {Towles}},\ }in\ \href {\doibase 10.1145/1654059.1654099}
  {\emph {\bibinfo {booktitle} {Proceedings of the Conference on High
  Performance Computing Networking, Storage and Analysis}}}\ (\bibinfo
  {publisher} {ACM},\ \bibinfo {year} {2009})\ pp.\ \bibinfo {pages}
  {39:1--39:11}\BibitemShut {NoStop}%
\bibitem [{\citenamefont {Shirts}\ and\ \citenamefont
  {Pande}(2000)}]{shirts2000science}%
  \BibitemOpen
  \bibfield  {author} {\bibinfo {author} {\bibfnamefont {M.}~\bibnamefont
  {Shirts}}\ and\ \bibinfo {author} {\bibfnamefont {V.~S.}\ \bibnamefont
  {Pande}},\ }\href {\doibase 10.1126/science.290.5498.1903} {\bibfield
  {journal} {\bibinfo  {journal} {Science}\ }\textbf {\bibinfo {volume}
  {290}},\ \bibinfo {pages} {1903} (\bibinfo {year} {2000})}\BibitemShut
  {NoStop}%
\bibitem [{\citenamefont {Buch}\ \emph {et~al.}(2010)\citenamefont {Buch},
  \citenamefont {Harvey}, \citenamefont {Giorgino}, \citenamefont {Anderson},\
  and\ \citenamefont {De~Fabritiis}}]{buch2010high}%
  \BibitemOpen
  \bibfield  {author} {\bibinfo {author} {\bibfnamefont {I.}~\bibnamefont
  {Buch}}, \bibinfo {author} {\bibfnamefont {M.~J.}\ \bibnamefont {Harvey}},
  \bibinfo {author} {\bibfnamefont {T.}~\bibnamefont {Giorgino}}, \bibinfo
  {author} {\bibfnamefont {D.}~\bibnamefont {Anderson}}, \ and\ \bibinfo
  {author} {\bibfnamefont {G.}~\bibnamefont {De~Fabritiis}},\ }\href {\doibase
  10.1021/ci900455r} {\bibfield  {journal} {\bibinfo  {journal} {J. Chem. Inf.
  Model.}\ }\textbf {\bibinfo {volume} {50}},\ \bibinfo {pages} {397} (\bibinfo
  {year} {2010})}\BibitemShut {NoStop}%
\bibitem [{\citenamefont {Kohlhoff}\ \emph {et~al.}(2014)\citenamefont
  {Kohlhoff}, \citenamefont {Shukla}, \citenamefont {Lawrenz}, \citenamefont
  {Bowman}, \citenamefont {Konerding}, \citenamefont {Belov}, \citenamefont
  {Altman},\ and\ \citenamefont {Pande}}]{kohlhoff2014cloud}%
  \BibitemOpen
  \bibfield  {author} {\bibinfo {author} {\bibfnamefont {K.~J.}\ \bibnamefont
  {Kohlhoff}}, \bibinfo {author} {\bibfnamefont {D.}~\bibnamefont {Shukla}},
  \bibinfo {author} {\bibfnamefont {M.}~\bibnamefont {Lawrenz}}, \bibinfo
  {author} {\bibfnamefont {G.~R.}\ \bibnamefont {Bowman}}, \bibinfo {author}
  {\bibfnamefont {D.~E.}\ \bibnamefont {Konerding}}, \bibinfo {author}
  {\bibfnamefont {D.}~\bibnamefont {Belov}}, \bibinfo {author} {\bibfnamefont
  {R.~B.}\ \bibnamefont {Altman}}, \ and\ \bibinfo {author} {\bibfnamefont
  {V.~S.}\ \bibnamefont {Pande}},\ }\href {\doibase 10.1038/nchem.1821}
  {\bibfield  {journal} {\bibinfo  {journal} {Nature Chem.}\ }\textbf {\bibinfo
  {volume} {6}},\ \bibinfo {pages} {15} (\bibinfo {year} {2014})}\BibitemShut
  {NoStop}%
\bibitem [{\citenamefont {Chodera}\ \emph {et~al.}(2007)\citenamefont
  {Chodera}, \citenamefont {Singhal}, \citenamefont {Pande}, \citenamefont
  {Dill},\ and\ \citenamefont {Swope}}]{chodera2007automatic}%
  \BibitemOpen
  \bibfield  {author} {\bibinfo {author} {\bibfnamefont {J.~D.}\ \bibnamefont
  {Chodera}}, \bibinfo {author} {\bibfnamefont {N.}~\bibnamefont {Singhal}},
  \bibinfo {author} {\bibfnamefont {V.~S.}\ \bibnamefont {Pande}}, \bibinfo
  {author} {\bibfnamefont {K.~A.}\ \bibnamefont {Dill}}, \ and\ \bibinfo
  {author} {\bibfnamefont {W.~C.}\ \bibnamefont {Swope}},\ }\href {\doibase
  10.1063/1.2714538} {\bibfield  {journal} {\bibinfo  {journal} {J. Chem.
  Phys.}\ }\textbf {\bibinfo {volume} {126}},\ \bibinfo {pages} {155101}
  (\bibinfo {year} {2007})}\BibitemShut {NoStop}%
\bibitem [{\citenamefont {Deuflhard}\ and\ \citenamefont
  {Weber}(2005)}]{deuflhard2005robust}%
  \BibitemOpen
  \bibfield  {author} {\bibinfo {author} {\bibfnamefont {P.}~\bibnamefont
  {Deuflhard}}\ and\ \bibinfo {author} {\bibfnamefont {M.}~\bibnamefont
  {Weber}},\ }\href {\doibase 10.1016/j.laa.2004.10.026} {\bibfield  {journal}
  {\bibinfo  {journal} {Linear Algebra Appl.}\ }\textbf {\bibinfo {volume}
  {398}},\ \bibinfo {pages} {161} (\bibinfo {year} {2005})}\BibitemShut
  {NoStop}%
\bibitem [{\citenamefont {Rohrdanz}\ \emph {et~al.}(2011)\citenamefont
  {Rohrdanz}, \citenamefont {Zheng}, \citenamefont {Maggioni},\ and\
  \citenamefont {Clementi}}]{rohrdanz2011determination}%
  \BibitemOpen
  \bibfield  {author} {\bibinfo {author} {\bibfnamefont {M.~A.}\ \bibnamefont
  {Rohrdanz}}, \bibinfo {author} {\bibfnamefont {W.}~\bibnamefont {Zheng}},
  \bibinfo {author} {\bibfnamefont {M.}~\bibnamefont {Maggioni}}, \ and\
  \bibinfo {author} {\bibfnamefont {C.}~\bibnamefont {Clementi}},\ }\href
  {\doibase http://dx.doi.org/10.1063/1.3569857} {\bibfield  {journal}
  {\bibinfo  {journal} {J. Chem. Phys.}\ }\textbf {\bibinfo {volume} {134}},\
  \bibinfo {pages} {124116} (\bibinfo {year} {2011})}\BibitemShut {NoStop}%
\bibitem [{\citenamefont {Altis}\ \emph {et~al.}(2007)\citenamefont {Altis},
  \citenamefont {Nguyen}, \citenamefont {Hegger},\ and\ \citenamefont
  {Stock}}]{altis2007dihedral}%
  \BibitemOpen
  \bibfield  {author} {\bibinfo {author} {\bibfnamefont {A.}~\bibnamefont
  {Altis}}, \bibinfo {author} {\bibfnamefont {P.~H.}\ \bibnamefont {Nguyen}},
  \bibinfo {author} {\bibfnamefont {R.}~\bibnamefont {Hegger}}, \ and\ \bibinfo
  {author} {\bibfnamefont {G.}~\bibnamefont {Stock}},\ }\href {\doibase
  http://dx.doi.org/10.1063/1.2746330} {\bibfield  {journal} {\bibinfo
  {journal} {J. Chem. Phys.}\ }\textbf {\bibinfo {volume} {126}},\ \bibinfo
  {pages} {244111} (\bibinfo {year} {2007})}\BibitemShut {NoStop}%
\bibitem [{\citenamefont {Das}\ \emph {et~al.}(2006)\citenamefont {Das},
  \citenamefont {Moll}, \citenamefont {Stamati}, \citenamefont {Kavraki},\ and\
  \citenamefont {Clementi}}]{das2006low}%
  \BibitemOpen
  \bibfield  {author} {\bibinfo {author} {\bibfnamefont {P.}~\bibnamefont
  {Das}}, \bibinfo {author} {\bibfnamefont {M.}~\bibnamefont {Moll}}, \bibinfo
  {author} {\bibfnamefont {H.}~\bibnamefont {Stamati}}, \bibinfo {author}
  {\bibfnamefont {L.~E.}\ \bibnamefont {Kavraki}}, \ and\ \bibinfo {author}
  {\bibfnamefont {C.}~\bibnamefont {Clementi}},\ }\href {\doibase
  10.1073/pnas.0603553103} {\bibfield  {journal} {\bibinfo  {journal} {Proc.
  Natl. Acad. Sci. U.S.A.}\ }\textbf {\bibinfo {volume} {103}},\ \bibinfo
  {pages} {9885} (\bibinfo {year} {2006})}\BibitemShut {NoStop}%
\bibitem [{\citenamefont {Krivov}\ and\ \citenamefont
  {Karplus}(2004)}]{krivov2004hidden}%
  \BibitemOpen
  \bibfield  {author} {\bibinfo {author} {\bibfnamefont {S.~V.}\ \bibnamefont
  {Krivov}}\ and\ \bibinfo {author} {\bibfnamefont {M.}~\bibnamefont
  {Karplus}},\ }\href {\doibase 10.1073/pnas.0406234101} {\bibfield  {journal}
  {\bibinfo  {journal} {Proc. Natl. Acad. Sci. U.S.A.}\ }\textbf {\bibinfo
  {volume} {101}},\ \bibinfo {pages} {14766} (\bibinfo {year}
  {2004})}\BibitemShut {NoStop}%
\bibitem [{\citenamefont {E}\ and\ \citenamefont
  {Vanden-Eijnden}(2006)}]{weinan2006towards}%
  \BibitemOpen
  \bibfield  {author} {\bibinfo {author} {\bibfnamefont {W.}~\bibnamefont {E}}\
  and\ \bibinfo {author} {\bibfnamefont {E.}~\bibnamefont {Vanden-Eijnden}},\
  }\href {\doibase 10.1007/s10955-005-9003-9} {\bibfield  {journal} {\bibinfo
  {journal} {J. Stat. Phys.}\ }\textbf {\bibinfo {volume} {123}},\ \bibinfo
  {pages} {503} (\bibinfo {year} {2006})}\BibitemShut {NoStop}%
\bibitem [{\citenamefont {McGibbon}, \citenamefont {Schwantes},\ and\
  \citenamefont {Pande}(2014)}]{McGibbon2014Statistical}%
  \BibitemOpen
  \bibfield  {author} {\bibinfo {author} {\bibfnamefont {R.~T.}\ \bibnamefont
  {McGibbon}}, \bibinfo {author} {\bibfnamefont {C.~R.}\ \bibnamefont
  {Schwantes}}, \ and\ \bibinfo {author} {\bibfnamefont {V.~S.}\ \bibnamefont
  {Pande}},\ }\href {\doibase 10.1021/jp411822r} {\bibfield  {journal}
  {\bibinfo  {journal} {J. Phys. Chem. B}\ }\textbf {\bibinfo {volume} {118}},\
  \bibinfo {pages} {6475} (\bibinfo {year} {2014})}\BibitemShut {NoStop}%
\bibitem [{\citenamefont {Vapnik}\ and\ \citenamefont
  {Vapnik}(1998)}]{vapnik1998statistical}%
  \BibitemOpen
  \bibfield  {author} {\bibinfo {author} {\bibfnamefont {V.~N.}\ \bibnamefont
  {Vapnik}}\ and\ \bibinfo {author} {\bibfnamefont {V.}~\bibnamefont
  {Vapnik}},\ }\href@noop {} {\emph {\bibinfo {title} {Statistical learning
  theory}}},\ Vol.~\bibinfo {volume} {2}\ (\bibinfo  {publisher} {Wiley New
  York},\ \bibinfo {year} {1998})\BibitemShut {NoStop}%
\bibitem [{\citenamefont {Larson}(1931)}]{larson1931shrinkage}%
  \BibitemOpen
  \bibfield  {author} {\bibinfo {author} {\bibfnamefont {S.~C.}\ \bibnamefont
  {Larson}},\ }\href {\doibase doi/10.1037/h0072400} {\bibfield  {journal}
  {\bibinfo  {journal} {J. Educ. Psychol.}\ }\textbf {\bibinfo {volume} {22}},\
  \bibinfo {pages} {45} (\bibinfo {year} {1931})}\BibitemShut {NoStop}%
\bibitem [{\citenamefont {Prinz}\ \emph {et~al.}(2011)\citenamefont {Prinz},
  \citenamefont {Wu}, \citenamefont {Sarich}, \citenamefont {Keller},
  \citenamefont {Senne}, \citenamefont {Held}, \citenamefont {Chodera},
  \citenamefont {Sch{\"u}tte},\ and\ \citenamefont
  {No{\'e}}}]{prinz2011markov}%
  \BibitemOpen
  \bibfield  {author} {\bibinfo {author} {\bibfnamefont {J.-H.}\ \bibnamefont
  {Prinz}}, \bibinfo {author} {\bibfnamefont {H.}~\bibnamefont {Wu}}, \bibinfo
  {author} {\bibfnamefont {M.}~\bibnamefont {Sarich}}, \bibinfo {author}
  {\bibfnamefont {B.}~\bibnamefont {Keller}}, \bibinfo {author} {\bibfnamefont
  {M.}~\bibnamefont {Senne}}, \bibinfo {author} {\bibfnamefont
  {M.}~\bibnamefont {Held}}, \bibinfo {author} {\bibfnamefont {J.~D.}\
  \bibnamefont {Chodera}}, \bibinfo {author} {\bibfnamefont {C.}~\bibnamefont
  {Sch{\"u}tte}}, \ and\ \bibinfo {author} {\bibfnamefont {F.}~\bibnamefont
  {No{\'e}}},\ }\href {\doibase 10.1063/1.3565032} {\bibfield  {journal}
  {\bibinfo  {journal} {J. Chem. Phys.}\ }\textbf {\bibinfo {volume} {134}},\
  \bibinfo {pages} {174105} (\bibinfo {year} {2011})}\BibitemShut {NoStop}%
\bibitem [{\citenamefont {Schwantes}\ and\ \citenamefont
  {Pande}(2013)}]{schwantes2013Improvements}%
  \BibitemOpen
  \bibfield  {author} {\bibinfo {author} {\bibfnamefont {C.~R.}\ \bibnamefont
  {Schwantes}}\ and\ \bibinfo {author} {\bibfnamefont {V.~S.}\ \bibnamefont
  {Pande}},\ }\href {\doibase 10.1021/ct300878a} {\bibfield  {journal}
  {\bibinfo  {journal} {J. Chem. Theory Comput.}\ }\textbf {\bibinfo {volume}
  {9}},\ \bibinfo {pages} {2000} (\bibinfo {year} {2013})}\BibitemShut
  {NoStop}%
\bibitem [{\citenamefont {P{\'e}rez-Hern{\'a}ndez}\ \emph
  {et~al.}(2013)\citenamefont {P{\'e}rez-Hern{\'a}ndez}, \citenamefont {Paul},
  \citenamefont {Giorgino}, \citenamefont {De~Fabritiis},\ and\ \citenamefont
  {No{\'e}}}]{perezhernandez2013Identification}%
  \BibitemOpen
  \bibfield  {author} {\bibinfo {author} {\bibfnamefont {G.}~\bibnamefont
  {P{\'e}rez-Hern{\'a}ndez}}, \bibinfo {author} {\bibfnamefont
  {F.}~\bibnamefont {Paul}}, \bibinfo {author} {\bibfnamefont {T.}~\bibnamefont
  {Giorgino}}, \bibinfo {author} {\bibfnamefont {G.}~\bibnamefont
  {De~Fabritiis}}, \ and\ \bibinfo {author} {\bibfnamefont {F.}~\bibnamefont
  {No{\'e}}},\ }\href {\doibase 10.1063/1.4811489} {\bibfield  {journal}
  {\bibinfo  {journal} {J. Chem. Phys.}\ }\textbf {\bibinfo {volume} {139}},\
  \bibinfo {eid} {015102} (\bibinfo {year} {2013})}\BibitemShut {NoStop}%
\bibitem [{\citenamefont {Maimon}\ and\ \citenamefont
  {Rokach}(2005)}]{maimon2005data}%
  \BibitemOpen
  \bibfield  {author} {\bibinfo {author} {\bibfnamefont {O.~Z.}\ \bibnamefont
  {Maimon}}\ and\ \bibinfo {author} {\bibfnamefont {L.}~\bibnamefont
  {Rokach}},\ }\href@noop {} {\emph {\bibinfo {title} {Data mining and
  knowledge discovery handbook}}},\ Vol.~\bibinfo {volume} {1}\ (\bibinfo
  {publisher} {Springer},\ \bibinfo {year} {2005})\BibitemShut {NoStop}%
\bibitem [{\citenamefont {Sch{\"u}tte}(1998)}]{schutte98conformational}%
  \BibitemOpen
  \bibfield  {author} {\bibinfo {author} {\bibfnamefont {C.}~\bibnamefont
  {Sch{\"u}tte}},\ }\href {http://publications.mi.fu-berlin.de/89/} {\enquote
  {\bibinfo {title} {Conformational dynamics: modelling, theory, algorithm, and
  application to biomolecules},}\ } (\bibinfo {year} {1998}),\ \bibinfo {note}
  {{H}abilitation Thesis, Free University Berlin}\BibitemShut {NoStop}%
\bibitem [{\citenamefont {Schmidt}(1907)}]{schmidt1907zur}%
  \BibitemOpen
  \bibfield  {author} {\bibinfo {author} {\bibfnamefont {E.}~\bibnamefont
  {Schmidt}},\ }\href {\doibase 10.1007/BF01449770} {\bibfield  {journal}
  {\bibinfo  {journal} {Math. Ann.}\ }\textbf {\bibinfo {volume} {63}},\
  \bibinfo {pages} {433} (\bibinfo {year} {1907})}\BibitemShut {NoStop}%
\bibitem [{\citenamefont {Courant}\ and\ \citenamefont
  {Hilbert}(2008)}]{courant2008methods}%
  \BibitemOpen
  \bibfield  {author} {\bibinfo {author} {\bibfnamefont {R.}~\bibnamefont
  {Courant}}\ and\ \bibinfo {author} {\bibfnamefont {D.}~\bibnamefont
  {Hilbert}},\ }\href {http://books.google.com/books?id=z3s4WnhEAbAC} {\emph
  {\bibinfo {title} {Methods of Mathematical Physics}}},\ \bibinfo {series}
  {Methods of Mathematical Physics}\ No.\ \bibinfo {number} {v. 1}\ (\bibinfo
  {publisher} {Wiley},\ \bibinfo {year} {2008})\BibitemShut {NoStop}%
\bibitem [{\citenamefont {Micchelli}\ and\ \citenamefont
  {Pinkus}(1978)}]{micchelli1971some}%
  \BibitemOpen
  \bibfield  {author} {\bibinfo {author} {\bibfnamefont {C.~A.}\ \bibnamefont
  {Micchelli}}\ and\ \bibinfo {author} {\bibfnamefont {A.}~\bibnamefont
  {Pinkus}},\ }\href {\doibase http://dx.doi.org/10.1016/0021-9045(78)90036-9}
  {\bibfield  {journal} {\bibinfo  {journal} {J. Approx. Theory}\ }\textbf
  {\bibinfo {volume} {24}},\ \bibinfo {pages} {51 } (\bibinfo {year}
  {1978})}\BibitemShut {NoStop}%
\bibitem [{\citenamefont {No{\'e}}\ and\ \citenamefont
  {Fischer}(2008)}]{noe2008transition}%
  \BibitemOpen
  \bibfield  {author} {\bibinfo {author} {\bibfnamefont {F.}~\bibnamefont
  {No{\'e}}}\ and\ \bibinfo {author} {\bibfnamefont {S.}~\bibnamefont
  {Fischer}},\ }\href {\doibase http://dx.doi.org/10.1016/j.sbi.2008.01.008}
  {\bibfield  {journal} {\bibinfo  {journal} {Curr. Opin. Struct. Biol.}\
  }\textbf {\bibinfo {volume} {18}},\ \bibinfo {pages} {154 } (\bibinfo {year}
  {2008})}\BibitemShut {NoStop}%
\bibitem [{\citenamefont {Bowman}\ \emph {et~al.}(2009)\citenamefont {Bowman},
  \citenamefont {Beauchamp}, \citenamefont {Boxer},\ and\ \citenamefont
  {Pande}}]{bowman2009progress}%
  \BibitemOpen
  \bibfield  {author} {\bibinfo {author} {\bibfnamefont {G.~R.}\ \bibnamefont
  {Bowman}}, \bibinfo {author} {\bibfnamefont {K.~A.}\ \bibnamefont
  {Beauchamp}}, \bibinfo {author} {\bibfnamefont {G.}~\bibnamefont {Boxer}}, \
  and\ \bibinfo {author} {\bibfnamefont {V.~S.}\ \bibnamefont {Pande}},\ }\href
  {\doibase http://dx.doi.org/10.1063/1.3216567} {\bibfield  {journal}
  {\bibinfo  {journal} {J. Chem. Phys.}\ }\textbf {\bibinfo {volume} {131}},\
  \bibinfo {eid} {124101} (\bibinfo {year} {2009})}\BibitemShut {NoStop}%
\bibitem [{\citenamefont {Pande}, \citenamefont {Beauchamp},\ and\
  \citenamefont {Bowman}(2010)}]{pande2010everything}%
  \BibitemOpen
  \bibfield  {author} {\bibinfo {author} {\bibfnamefont {V.~S.}\ \bibnamefont
  {Pande}}, \bibinfo {author} {\bibfnamefont {K.}~\bibnamefont {Beauchamp}}, \
  and\ \bibinfo {author} {\bibfnamefont {G.~R.}\ \bibnamefont {Bowman}},\
  }\href {\doibase http://dx.doi.org/10.1016/j.ymeth.2010.06.002} {\bibfield
  {journal} {\bibinfo  {journal} {Methods}\ }\textbf {\bibinfo {volume} {52}},\
  \bibinfo {pages} {99 } (\bibinfo {year} {2010})}\BibitemShut {NoStop}%
\bibitem [{\citenamefont {Chodera}\ and\ \citenamefont
  {No{\'e}}(2014)}]{chodera2014markov}%
  \BibitemOpen
  \bibfield  {author} {\bibinfo {author} {\bibfnamefont {J.~D.}\ \bibnamefont
  {Chodera}}\ and\ \bibinfo {author} {\bibfnamefont {F.~F.}\ \bibnamefont
  {No{\'e}}},\ }\href {\doibase http://dx.doi.org/10.1016/j.sbi.2014.04.002}
  {\bibfield  {journal} {\bibinfo  {journal} {Curr. Opin. Stuct. Biol}\
  }\textbf {\bibinfo {volume} {25}},\ \bibinfo {pages} {135 } (\bibinfo {year}
  {2014})}\BibitemShut {NoStop}%
\bibitem [{\citenamefont {No{\'e}}\ and\ \citenamefont
  {N\"{u}ske}(2013)}]{noe2013variational}%
  \BibitemOpen
  \bibfield  {author} {\bibinfo {author} {\bibfnamefont {F.}~\bibnamefont
  {No{\'e}}}\ and\ \bibinfo {author} {\bibfnamefont {F.}~\bibnamefont
  {N\"{u}ske}},\ }\href {\doibase 10.1137/110858616} {\bibfield  {journal}
  {\bibinfo  {journal} {Multiscale Model. Simul.}\ }\textbf {\bibinfo {volume}
  {11}},\ \bibinfo {pages} {635} (\bibinfo {year} {2013})}\BibitemShut
  {NoStop}%
\bibitem [{\citenamefont {N{\"u}ske}\ \emph {et~al.}(2014)\citenamefont
  {N{\"u}ske}, \citenamefont {Keller}, \citenamefont {P{\'e}rez-Hern{\'a}ndez},
  \citenamefont {Mey},\ and\ \citenamefont {No{\'e}}}]{nuske2014variational}%
  \BibitemOpen
  \bibfield  {author} {\bibinfo {author} {\bibfnamefont {F.}~\bibnamefont
  {N{\"u}ske}}, \bibinfo {author} {\bibfnamefont {B.~G.}\ \bibnamefont
  {Keller}}, \bibinfo {author} {\bibfnamefont {G.}~\bibnamefont
  {P{\'e}rez-Hern{\'a}ndez}}, \bibinfo {author} {\bibfnamefont {A.~S. J.~S.}\
  \bibnamefont {Mey}}, \ and\ \bibinfo {author} {\bibfnamefont
  {F.}~\bibnamefont {No{\'e}}},\ }\href {\doibase 10.1021/ct4009156} {\bibfield
   {journal} {\bibinfo  {journal} {J. Chem. Theory Comput.}\ }\textbf {\bibinfo
  {volume} {10}},\ \bibinfo {pages} {1739} (\bibinfo {year}
  {2014})}\BibitemShut {NoStop}%
\bibitem [{\citenamefont {Fan}(1949)}]{fan1949theorem}%
  \BibitemOpen
  \bibfield  {author} {\bibinfo {author} {\bibfnamefont {K.}~\bibnamefont
  {Fan}},\ }\href {http://www.pnas.org/content/35/11/652.short} {\bibfield
  {journal} {\bibinfo  {journal} {Proc. Natl. Acad. Sci. U.S.A.}\ }\textbf
  {\bibinfo {volume} {35}},\ \bibinfo {pages} {652} (\bibinfo {year}
  {1949})}\BibitemShut {NoStop}%
\bibitem [{\citenamefont {Overton}\ and\ \citenamefont
  {Womersley}(1992)}]{overton1992sum}%
  \BibitemOpen
  \bibfield  {author} {\bibinfo {author} {\bibfnamefont {M.~L.}\ \bibnamefont
  {Overton}}\ and\ \bibinfo {author} {\bibfnamefont {R.~S.}\ \bibnamefont
  {Womersley}},\ }\href {\doibase 10.1137/0613006} {\bibfield  {journal}
  {\bibinfo  {journal} {SIAM J. Matrix Anal. Appl.}\ }\textbf {\bibinfo
  {volume} {13}},\ \bibinfo {pages} {41} (\bibinfo {year} {1992})}\BibitemShut
  {NoStop}%
\bibitem [{\citenamefont {Absil}\ \emph {et~al.}(2002)\citenamefont {Absil},
  \citenamefont {Mahony}, \citenamefont {Sepulchre},\ and\ \citenamefont
  {Van~Dooren}}]{absil2002grassmann}%
  \BibitemOpen
  \bibfield  {author} {\bibinfo {author} {\bibfnamefont {P.-A.}\ \bibnamefont
  {Absil}}, \bibinfo {author} {\bibfnamefont {R.}~\bibnamefont {Mahony}},
  \bibinfo {author} {\bibfnamefont {R.}~\bibnamefont {Sepulchre}}, \ and\
  \bibinfo {author} {\bibfnamefont {P.}~\bibnamefont {Van~Dooren}},\ }\href
  {\doibase 10.1137/S0036144500378648} {\bibfield  {journal} {\bibinfo
  {journal} {SIAM Rev.}\ }\textbf {\bibinfo {volume} {44}},\ \bibinfo {pages}
  {57} (\bibinfo {year} {2002})}\BibitemShut {NoStop}%
\bibitem [{\citenamefont {Sezer}, \citenamefont {Freed},\ and\ \citenamefont
  {Roux}(2008)}]{sezer2008using}%
  \BibitemOpen
  \bibfield  {author} {\bibinfo {author} {\bibfnamefont {D.}~\bibnamefont
  {Sezer}}, \bibinfo {author} {\bibfnamefont {J.~H.}\ \bibnamefont {Freed}}, \
  and\ \bibinfo {author} {\bibfnamefont {B.}~\bibnamefont {Roux}},\ }\href
  {\doibase 10.1021/jp801608v} {\bibfield  {journal} {\bibinfo  {journal} {J.
  Phys. Chem. B}\ }\textbf {\bibinfo {volume} {112}},\ \bibinfo {pages} {11014}
  (\bibinfo {year} {2008})}\BibitemShut {NoStop}%
\bibitem [{\citenamefont {Muff}\ and\ \citenamefont
  {Caflisch}(2009)}]{muff2009identification}%
  \BibitemOpen
  \bibfield  {author} {\bibinfo {author} {\bibfnamefont {S.}~\bibnamefont
  {Muff}}\ and\ \bibinfo {author} {\bibfnamefont {A.}~\bibnamefont
  {Caflisch}},\ }\href {\doibase http://dx.doi.org/10.1063/1.3099705}
  {\bibfield  {journal} {\bibinfo  {journal} {J. Chem. Phys.}\ }\textbf
  {\bibinfo {volume} {130}},\ \bibinfo {eid} {125104} (\bibinfo {year}
  {2009})}\BibitemShut {NoStop}%
\bibitem [{\citenamefont {Buch}, \citenamefont {Giorgino},\ and\ \citenamefont
  {De~Fabritiis}(2011)}]{buch2011complete}%
  \BibitemOpen
  \bibfield  {author} {\bibinfo {author} {\bibfnamefont {I.}~\bibnamefont
  {Buch}}, \bibinfo {author} {\bibfnamefont {T.}~\bibnamefont {Giorgino}}, \
  and\ \bibinfo {author} {\bibfnamefont {G.}~\bibnamefont {De~Fabritiis}},\
  }\href {\doibase 10.1073/pnas.1103547108} {\bibfield  {journal} {\bibinfo
  {journal} {Proc. Natl. Acad. Sci. U.S.A.}\ }\textbf {\bibinfo {volume}
  {108}},\ \bibinfo {pages} {10184} (\bibinfo {year} {2011})}\BibitemShut
  {NoStop}%
\bibitem [{\citenamefont {Voelz}\ \emph {et~al.}(2010)\citenamefont {Voelz},
  \citenamefont {Bowman}, \citenamefont {Beauchamp},\ and\ \citenamefont
  {Pande}}]{voelz2010molecular}%
  \BibitemOpen
  \bibfield  {author} {\bibinfo {author} {\bibfnamefont {V.~A.}\ \bibnamefont
  {Voelz}}, \bibinfo {author} {\bibfnamefont {G.~R.}\ \bibnamefont {Bowman}},
  \bibinfo {author} {\bibfnamefont {K.}~\bibnamefont {Beauchamp}}, \ and\
  \bibinfo {author} {\bibfnamefont {V.~S.}\ \bibnamefont {Pande}},\ }\href
  {\doibase 10.1021/ja9090353} {\bibfield  {journal} {\bibinfo  {journal} {J.
  Am. Chem. Soc.}\ }\textbf {\bibinfo {volume} {132}},\ \bibinfo {pages} {1526}
  (\bibinfo {year} {2010})}\BibitemShut {NoStop}%
\bibitem [{\citenamefont {Beauchamp}\ \emph {et~al.}(2012)\citenamefont
  {Beauchamp}, \citenamefont {McGibbon}, \citenamefont {Lin},\ and\
  \citenamefont {Pande}}]{Beauchamp2012simple}%
  \BibitemOpen
  \bibfield  {author} {\bibinfo {author} {\bibfnamefont {K.~A.}\ \bibnamefont
  {Beauchamp}}, \bibinfo {author} {\bibfnamefont {R.~T.}\ \bibnamefont
  {McGibbon}}, \bibinfo {author} {\bibfnamefont {Y.-S.}\ \bibnamefont {Lin}}, \
  and\ \bibinfo {author} {\bibfnamefont {V.~S.}\ \bibnamefont {Pande}},\ }\href
  {\doibase 10.1073/pnas.1201810109} {\bibfield  {journal} {\bibinfo  {journal}
  {Proc. Natl. Acad. Sci. U.S.A.}\ }\textbf {\bibinfo {volume} {109}},\
  \bibinfo {pages} {17807} (\bibinfo {year} {2012})}\BibitemShut {NoStop}%
\bibitem [{\citenamefont {Zhuang}\ \emph {et~al.}(2011)\citenamefont {Zhuang},
  \citenamefont {Cui}, \citenamefont {Silva},\ and\ \citenamefont
  {Huang}}]{zhuang2011simulating}%
  \BibitemOpen
  \bibfield  {author} {\bibinfo {author} {\bibfnamefont {W.}~\bibnamefont
  {Zhuang}}, \bibinfo {author} {\bibfnamefont {R.~Z.}\ \bibnamefont {Cui}},
  \bibinfo {author} {\bibfnamefont {D.-A.}\ \bibnamefont {Silva}}, \ and\
  \bibinfo {author} {\bibfnamefont {X.}~\bibnamefont {Huang}},\ }\href
  {\doibase 10.1021/jp109592b} {\bibfield  {journal} {\bibinfo  {journal} {J.
  Phys. Chem. B}\ }\textbf {\bibinfo {volume} {115}},\ \bibinfo {pages} {5415}
  (\bibinfo {year} {2011})}\BibitemShut {NoStop}%
\bibitem [{\citenamefont {Sadiq}, \citenamefont {Noé},\ and\ \citenamefont
  {De~Fabritiis}(2012)}]{sadiq2012kinetic}%
  \BibitemOpen
  \bibfield  {author} {\bibinfo {author} {\bibfnamefont {S.~K.}\ \bibnamefont
  {Sadiq}}, \bibinfo {author} {\bibfnamefont {F.}~\bibnamefont {Noé}}, \ and\
  \bibinfo {author} {\bibfnamefont {G.}~\bibnamefont {De~Fabritiis}},\ }\href
  {\doibase 10.1073/pnas.1210983109} {\bibfield  {journal} {\bibinfo  {journal}
  {Proc. Natl. Acad. Sci. U.S.A.}\ }\textbf {\bibinfo {volume} {109}},\
  \bibinfo {pages} {20449} (\bibinfo {year} {2012})}\BibitemShut {NoStop}%
\bibitem [{\citenamefont {Choudhary}\ \emph {et~al.}(2014)\citenamefont
  {Choudhary}, \citenamefont {Paz}, \citenamefont {Adelman}, \citenamefont
  {Colletier}, \citenamefont {Abramson},\ and\ \citenamefont
  {Grabe}}]{choudhary2014structure}%
  \BibitemOpen
  \bibfield  {author} {\bibinfo {author} {\bibfnamefont {O.~P.}\ \bibnamefont
  {Choudhary}}, \bibinfo {author} {\bibfnamefont {A.}~\bibnamefont {Paz}},
  \bibinfo {author} {\bibfnamefont {J.~L.}\ \bibnamefont {Adelman}}, \bibinfo
  {author} {\bibfnamefont {J.-P.}\ \bibnamefont {Colletier}}, \bibinfo {author}
  {\bibfnamefont {J.}~\bibnamefont {Abramson}}, \ and\ \bibinfo {author}
  {\bibfnamefont {M.}~\bibnamefont {Grabe}},\ }\href {\doibase
  10.1038/nsmb.2841} {\bibfield  {journal} {\bibinfo  {journal} {Nat. Struct.
  Mol. Biol.}\ }\textbf {\bibinfo {volume} {21}},\ \bibinfo {pages} {626}
  (\bibinfo {year} {2014})}\BibitemShut {NoStop}%
\bibitem [{\citenamefont {Shukla}\ \emph {et~al.}(2014)\citenamefont {Shukla},
  \citenamefont {Meng}, \citenamefont {Roux},\ and\ \citenamefont
  {Pande}}]{shukla2014activation}%
  \BibitemOpen
  \bibfield  {author} {\bibinfo {author} {\bibfnamefont {D.}~\bibnamefont
  {Shukla}}, \bibinfo {author} {\bibfnamefont {Y.}~\bibnamefont {Meng}},
  \bibinfo {author} {\bibfnamefont {B.}~\bibnamefont {Roux}}, \ and\ \bibinfo
  {author} {\bibfnamefont {V.~S.}\ \bibnamefont {Pande}},\ }\href {\doibase
  10.1038/ncomms4397} {\bibfield  {journal} {\bibinfo  {journal} {Nat.
  Commun.}\ }\textbf {\bibinfo {volume} {5}},\  (\bibinfo {year}
  {2014})}\BibitemShut {NoStop}%
\bibitem [{\citenamefont {Anderson}\ \emph {et~al.}(1999)\citenamefont
  {Anderson}, \citenamefont {Bai}, \citenamefont {Bischof}, \citenamefont
  {Blackford}, \citenamefont {Demmel}, \citenamefont {Dongarra}, \citenamefont
  {Du~Croz}, \citenamefont {Greenbaum}, \citenamefont {Hammarling},
  \citenamefont {McKenney},\ and\ \citenamefont {Sorensen}}]{lapack99}%
  \BibitemOpen
  \bibfield  {author} {\bibinfo {author} {\bibfnamefont {E.}~\bibnamefont
  {Anderson}}, \bibinfo {author} {\bibfnamefont {Z.}~\bibnamefont {Bai}},
  \bibinfo {author} {\bibfnamefont {C.}~\bibnamefont {Bischof}}, \bibinfo
  {author} {\bibfnamefont {S.}~\bibnamefont {Blackford}}, \bibinfo {author}
  {\bibfnamefont {J.}~\bibnamefont {Demmel}}, \bibinfo {author} {\bibfnamefont
  {J.}~\bibnamefont {Dongarra}}, \bibinfo {author} {\bibfnamefont
  {J.}~\bibnamefont {Du~Croz}}, \bibinfo {author} {\bibfnamefont
  {A.}~\bibnamefont {Greenbaum}}, \bibinfo {author} {\bibfnamefont
  {S.}~\bibnamefont {Hammarling}}, \bibinfo {author} {\bibfnamefont
  {A.}~\bibnamefont {McKenney}}, \ and\ \bibinfo {author} {\bibfnamefont
  {D.}~\bibnamefont {Sorensen}},\ }\href@noop {} {\emph {\bibinfo {title}
  {{LAPACK} Users' Guide}}},\ \bibinfo {edition} {3rd}\ ed.\ (\bibinfo
  {publisher} {Society for Industrial and Applied Mathematics},\ \bibinfo
  {address} {Philadelphia, PA},\ \bibinfo {year} {1999})\BibitemShut {NoStop}%
\bibitem [{\citenamefont {Sch{\"o}lkopf}, \citenamefont {Smola},\ and\
  \citenamefont {M{\"u}ller}(1998)}]{scholkopf1998nonlinear}%
  \BibitemOpen
  \bibfield  {author} {\bibinfo {author} {\bibfnamefont {B.}~\bibnamefont
  {Sch{\"o}lkopf}}, \bibinfo {author} {\bibfnamefont {A.}~\bibnamefont
  {Smola}}, \ and\ \bibinfo {author} {\bibfnamefont {K.-R.}\ \bibnamefont
  {M{\"u}ller}},\ }\href {\doibase 10.1162/089976698300017467} {\bibfield
  {journal} {\bibinfo  {journal} {Neural Comput.}\ }\textbf {\bibinfo {volume}
  {10}},\ \bibinfo {pages} {1299} (\bibinfo {year} {1998})}\BibitemShut
  {NoStop}%
\bibitem [{\citenamefont {Kjems}, \citenamefont {Hansen},\ and\ \citenamefont
  {Strother}(2000)}]{kjems2000generalizable}%
  \BibitemOpen
  \bibfield  {author} {\bibinfo {author} {\bibfnamefont {U.}~\bibnamefont
  {Kjems}}, \bibinfo {author} {\bibfnamefont {L.~K.}\ \bibnamefont {Hansen}}, \
  and\ \bibinfo {author} {\bibfnamefont {S.~C.}\ \bibnamefont {Strother}},\
  }in\ \href
  {http://papers.nips.cc/paper/1881-generalizable-singular-value-decomposition-for-ill-posed-datasets}
  {\emph {\bibinfo {booktitle} {Advances in Neural Information Processing
  Systems 13}}}\ (\bibinfo  {publisher} {MIT Press},\ \bibinfo {year} {2000})\
  pp.\ \bibinfo {pages} {549--555}\BibitemShut {NoStop}%
\bibitem [{\citenamefont {Abrahamsen}\ and\ \citenamefont
  {Hansen}(2011)}]{abrahamsen2011cure}%
  \BibitemOpen
  \bibfield  {author} {\bibinfo {author} {\bibfnamefont {T.~J.}\ \bibnamefont
  {Abrahamsen}}\ and\ \bibinfo {author} {\bibfnamefont {L.~K.}\ \bibnamefont
  {Hansen}},\ }\href {http://jmlr.org/papers/v12/abrahamsen11a.html} {\bibfield
   {journal} {\bibinfo  {journal} {J. Mach. Learn. Res.}\ }\textbf {\bibinfo
  {volume} {12}},\ \bibinfo {pages} {2027} (\bibinfo {year}
  {2011})}\BibitemShut {NoStop}%
\bibitem [{\citenamefont {Cornec}(2010)}]{cornec2010concentration}%
  \BibitemOpen
  \bibfield  {author} {\bibinfo {author} {\bibfnamefont {M.}~\bibnamefont
  {Cornec}},\ }\href {http://arxiv.org/abs/1011.0096} {\enquote {\bibinfo
  {title} {Concentration inequalities of the cross-validation estimator for
  empirical risk minimiser},}\ } (\bibinfo {year} {2010}),\ \Eprint
  {http://arxiv.org/abs/1011.0096} {arXiv:1011.0096 [stat]} \BibitemShut
  {NoStop}%
\bibitem [{\citenamefont {Park}\ and\ \citenamefont
  {Pande}(2006)}]{park2006validation}%
  \BibitemOpen
  \bibfield  {author} {\bibinfo {author} {\bibfnamefont {S.}~\bibnamefont
  {Park}}\ and\ \bibinfo {author} {\bibfnamefont {V.~S.}\ \bibnamefont
  {Pande}},\ }\href@noop {} {\bibfield  {journal} {\bibinfo  {journal} {J Chem.
  Phys.}\ }\textbf {\bibinfo {volume} {124}},\ \bibinfo {pages} {054118}
  (\bibinfo {year} {2006})}\BibitemShut {NoStop}%
\bibitem [{\citenamefont {Sarich}, \citenamefont {No\'{e}},\ and\ \citenamefont
  {Sch\"{u}tte}(2010)}]{sarich2010approximation}%
  \BibitemOpen
  \bibfield  {author} {\bibinfo {author} {\bibfnamefont {M.}~\bibnamefont
  {Sarich}}, \bibinfo {author} {\bibfnamefont {F.}~\bibnamefont {No\'{e}}}, \
  and\ \bibinfo {author} {\bibfnamefont {C.}~\bibnamefont {Sch\"{u}tte}},\
  }\href {\doibase 10.1137/090764049} {\bibfield  {journal} {\bibinfo
  {journal} {Multiscale Model. Simul.}\ }\textbf {\bibinfo {volume} {8}},\
  \bibinfo {pages} {1154} (\bibinfo {year} {2010})}\BibitemShut {NoStop}%
\bibitem [{\citenamefont {Zhou}\ and\ \citenamefont
  {Caflisch}(2012)}]{zhou2012distribution}%
  \BibitemOpen
  \bibfield  {author} {\bibinfo {author} {\bibfnamefont {T.}~\bibnamefont
  {Zhou}}\ and\ \bibinfo {author} {\bibfnamefont {A.}~\bibnamefont
  {Caflisch}},\ }\href {\doibase 10.1021/ct3003145} {\bibfield  {journal}
  {\bibinfo  {journal} {J. Chem. Theory Comput.}\ }\textbf {\bibinfo {volume}
  {8}},\ \bibinfo {pages} {2930} (\bibinfo {year} {2012})}\BibitemShut
  {NoStop}%
\bibitem [{\citenamefont {Theobald}(2005)}]{theobald2005rapid}%
  \BibitemOpen
  \bibfield  {author} {\bibinfo {author} {\bibfnamefont {D.~L.}\ \bibnamefont
  {Theobald}},\ }\href {\doibase 10.1107/S0108767305015266} {\bibfield
  {journal} {\bibinfo  {journal} {Acta Crystallogr., Sect. A: Found.
  Crystallogr.}\ }\textbf {\bibinfo {volume} {61}},\ \bibinfo {pages} {478}
  (\bibinfo {year} {2005})}\BibitemShut {NoStop}%
\bibitem [{\citenamefont {Beauchamp}\ \emph {et~al.}(2011)\citenamefont
  {Beauchamp}, \citenamefont {Bowman}, \citenamefont {Lane}, \citenamefont
  {Maibaum}, \citenamefont {Haque},\ and\ \citenamefont
  {Pande}}]{Beauchamp2011Msmbuilder2}%
  \BibitemOpen
  \bibfield  {author} {\bibinfo {author} {\bibfnamefont {K.~A.}\ \bibnamefont
  {Beauchamp}}, \bibinfo {author} {\bibfnamefont {G.~R.}\ \bibnamefont
  {Bowman}}, \bibinfo {author} {\bibfnamefont {T.~J.}\ \bibnamefont {Lane}},
  \bibinfo {author} {\bibfnamefont {L.}~\bibnamefont {Maibaum}}, \bibinfo
  {author} {\bibfnamefont {I.~S.}\ \bibnamefont {Haque}}, \ and\ \bibinfo
  {author} {\bibfnamefont {V.~S.}\ \bibnamefont {Pande}},\ }\href {\doibase
  10.1021/ct200463m} {\bibfield  {journal} {\bibinfo  {journal} {J. Chem.
  Theory Comput.}\ }\textbf {\bibinfo {volume} {7}},\ \bibinfo {pages} {3412}
  (\bibinfo {year} {2011})}\BibitemShut {NoStop}%
\bibitem [{\citenamefont {Lloyd}(1982)}]{Lloyd_1982}%
  \BibitemOpen
  \bibfield  {author} {\bibinfo {author} {\bibfnamefont {S.}~\bibnamefont
  {Lloyd}},\ }\href {\doibase 10.1109/tit.1982.1056489} {\bibfield  {journal}
  {\bibinfo  {journal} {IEEE Trans. Inform. Theory}\ }\textbf {\bibinfo
  {volume} {28}},\ \bibinfo {pages} {129–137} (\bibinfo {year}
  {1982})}\BibitemShut {NoStop}%
\bibitem [{\citenamefont {Senne}\ \emph {et~al.}(2012)\citenamefont {Senne},
  \citenamefont {Trendelkamp-Schroer}, \citenamefont {Mey}, \citenamefont
  {Sch\"{u}tte},\ and\ \citenamefont {No\'e}}]{senne2012emma}%
  \BibitemOpen
  \bibfield  {author} {\bibinfo {author} {\bibfnamefont {M.}~\bibnamefont
  {Senne}}, \bibinfo {author} {\bibfnamefont {B.}~\bibnamefont
  {Trendelkamp-Schroer}}, \bibinfo {author} {\bibfnamefont {A.~S.}\
  \bibnamefont {Mey}}, \bibinfo {author} {\bibfnamefont {C.}~\bibnamefont
  {Sch\"{u}tte}}, \ and\ \bibinfo {author} {\bibfnamefont {F.}~\bibnamefont
  {No\'e}},\ }\href {\doibase 10.1021/ct300274u} {\bibfield  {journal}
  {\bibinfo  {journal} {J. Chem. Theory Comput.}\ }\textbf {\bibinfo {volume}
  {8}},\ \bibinfo {pages} {2223} (\bibinfo {year} {2012})}\BibitemShut
  {NoStop}%
\bibitem [{\citenamefont {Theophilou}(1979)}]{theophilou79energy}%
  \BibitemOpen
  \bibfield  {author} {\bibinfo {author} {\bibfnamefont {A.~K.}\ \bibnamefont
  {Theophilou}},\ }\href {http://stacks.iop.org/0022-3719/12/i=24/a=013}
  {\bibfield  {journal} {\bibinfo  {journal} {J. Phys. C}\ }\textbf {\bibinfo
  {volume} {12}},\ \bibinfo {pages} {5419} (\bibinfo {year}
  {1979})}\BibitemShut {NoStop}%
\bibitem [{\citenamefont {Gross}, \citenamefont {Oliveira},\ and\ \citenamefont
  {Kohn}(1988)}]{gross88rayleigh}%
  \BibitemOpen
  \bibfield  {author} {\bibinfo {author} {\bibfnamefont {E.~K.~U.}\
  \bibnamefont {Gross}}, \bibinfo {author} {\bibfnamefont {L.~N.}\ \bibnamefont
  {Oliveira}}, \ and\ \bibinfo {author} {\bibfnamefont {W.}~\bibnamefont
  {Kohn}},\ }\href {\doibase 10.1103/PhysRevA.37.2805} {\bibfield  {journal}
  {\bibinfo  {journal} {Phys. Rev. A}\ }\textbf {\bibinfo {volume} {37}},\
  \bibinfo {pages} {2805} (\bibinfo {year} {1988})}\BibitemShut {NoStop}%
\bibitem [{\citenamefont {Gidopoulos}, \citenamefont {Papaconstantinou},\ and\
  \citenamefont {Gross}(2002)}]{gidopoulos2002ensemble}%
  \BibitemOpen
  \bibfield  {author} {\bibinfo {author} {\bibfnamefont {N.}~\bibnamefont
  {Gidopoulos}}, \bibinfo {author} {\bibfnamefont {P.}~\bibnamefont
  {Papaconstantinou}}, \ and\ \bibinfo {author} {\bibfnamefont
  {E.}~\bibnamefont {Gross}},\ }\href {\doibase
  http://dx.doi.org/10.1016/S0921-4526(02)00799-8} {\bibfield  {journal}
  {\bibinfo  {journal} {Physica B}\ }\textbf {\bibinfo {volume} {318}},\
  \bibinfo {pages} {328 } (\bibinfo {year} {2002})}\BibitemShut {NoStop}%
\bibitem [{\citenamefont {{Lai}}, \citenamefont {{Lu}},\ and\ \citenamefont
  {{Osher}}(2014)}]{lai2014density}%
  \BibitemOpen
  \bibfield  {author} {\bibinfo {author} {\bibfnamefont {R.}~\bibnamefont
  {{Lai}}}, \bibinfo {author} {\bibfnamefont {J.}~\bibnamefont {{Lu}}}, \ and\
  \bibinfo {author} {\bibfnamefont {S.}~\bibnamefont {{Osher}}},\ }\href@noop
  {} {\enquote {\bibinfo {title} {{Density matrix minimization with $\ell_1$
  regularization}},}\ } (\bibinfo {year} {2014}),\ \Eprint
  {http://arxiv.org/abs/1403.1525} {arXiv:1403.1525 [math-ph]} \BibitemShut
  {NoStop}%
\bibitem [{\citenamefont {Rao}(1948)}]{rao1948utilization}%
  \BibitemOpen
  \bibfield  {author} {\bibinfo {author} {\bibfnamefont {C.~R.}\ \bibnamefont
  {Rao}},\ }\href {http://www.jstor.org/stable/2983775} {\bibfield  {journal}
  {\bibinfo  {journal} {J. R. Stat. Soc. Ser. B Stat. Methodol.}\ }\textbf
  {\bibinfo {volume} {10}},\ \bibinfo {pages} {pp. 159} (\bibinfo {year}
  {1948})}\BibitemShut {NoStop}%
\bibitem [{\citenamefont {Baudat}\ and\ \citenamefont
  {Anouar}(2000)}]{baudat2000generalized}%
  \BibitemOpen
  \bibfield  {author} {\bibinfo {author} {\bibfnamefont {G.}~\bibnamefont
  {Baudat}}\ and\ \bibinfo {author} {\bibfnamefont {F.}~\bibnamefont
  {Anouar}},\ }\href {\doibase 10.1162/089976600300014980} {\bibfield
  {journal} {\bibinfo  {journal} {Neural Comput.}\ }\textbf {\bibinfo {volume}
  {12}},\ \bibinfo {pages} {2385} (\bibinfo {year} {2000})}\BibitemShut
  {NoStop}%
\bibitem [{\citenamefont {Clemmensen}\ \emph {et~al.}(2011)\citenamefont
  {Clemmensen}, \citenamefont {Hastie}, \citenamefont {Witten},\ and\
  \citenamefont {Ersb{\o}ll}}]{clemmensen2011sparse}%
  \BibitemOpen
  \bibfield  {author} {\bibinfo {author} {\bibfnamefont {L.}~\bibnamefont
  {Clemmensen}}, \bibinfo {author} {\bibfnamefont {T.}~\bibnamefont {Hastie}},
  \bibinfo {author} {\bibfnamefont {D.}~\bibnamefont {Witten}}, \ and\ \bibinfo
  {author} {\bibfnamefont {B.}~\bibnamefont {Ersb{\o}ll}},\ }\href {\doibase
  10.1198/TECH.2011.08118} {\bibfield  {journal} {\bibinfo  {journal}
  {Technometrics}\ }\textbf {\bibinfo {volume} {53}},\  (\bibinfo {year}
  {2011})}\BibitemShut {NoStop}%
\bibitem [{\citenamefont {Sriperumbudur}, \citenamefont {Torres},\ and\
  \citenamefont {Lanckriet}(2011)}]{sriperumbudur2011majorizarion}%
  \BibitemOpen
  \bibfield  {author} {\bibinfo {author} {\bibfnamefont {B.~K.}\ \bibnamefont
  {Sriperumbudur}}, \bibinfo {author} {\bibfnamefont {D.}~\bibnamefont
  {Torres}}, \ and\ \bibinfo {author} {\bibfnamefont {G.}~\bibnamefont
  {Lanckriet}},\ }\href {\doibase 10.1007/s10994-010-5226-3} {\bibfield
  {journal} {\bibinfo  {journal} {Mach. Learn.}\ }\textbf {\bibinfo {volume}
  {85}},\ \bibinfo {pages} {3} (\bibinfo {year} {2011})}\BibitemShut {NoStop}%
\bibitem [{\citenamefont {Kellogg}, \citenamefont {Lange},\ and\ \citenamefont
  {Baker}(2012)}]{kellogg2012evaluation}%
  \BibitemOpen
  \bibfield  {author} {\bibinfo {author} {\bibfnamefont {E.~H.}\ \bibnamefont
  {Kellogg}}, \bibinfo {author} {\bibfnamefont {O.~F.}\ \bibnamefont {Lange}},
  \ and\ \bibinfo {author} {\bibfnamefont {D.}~\bibnamefont {Baker}},\ }\href
  {\doibase 10.1021/jp3044303} {\bibfield  {journal} {\bibinfo  {journal} {J.
  Phys. Chem. B}\ }\textbf {\bibinfo {volume} {116}},\ \bibinfo {pages} {11405}
  (\bibinfo {year} {2012})}\BibitemShut {NoStop}%
\bibitem [{\citenamefont {Rains}\ and\ \citenamefont
  {Andersen}(2010)}]{rains2010bayesian}%
  \BibitemOpen
  \bibfield  {author} {\bibinfo {author} {\bibfnamefont {E.~K.}\ \bibnamefont
  {Rains}}\ and\ \bibinfo {author} {\bibfnamefont {H.~C.}\ \bibnamefont
  {Andersen}},\ }\href {\doibase http://dx.doi.org/10.1063/1.3496438}
  {\bibfield  {journal} {\bibinfo  {journal} {J. Chem. Phys.}\ }\textbf
  {\bibinfo {volume} {133}},\ \bibinfo {eid} {144113} (\bibinfo {year}
  {2010})}\BibitemShut {NoStop}%
\bibitem [{\citenamefont {Bacallado}, \citenamefont {Chodera},\ and\
  \citenamefont {Pande}(2009)}]{bacallado2009bayesian}%
  \BibitemOpen
  \bibfield  {author} {\bibinfo {author} {\bibfnamefont {S.}~\bibnamefont
  {Bacallado}}, \bibinfo {author} {\bibfnamefont {J.~D.}\ \bibnamefont
  {Chodera}}, \ and\ \bibinfo {author} {\bibfnamefont {V.}~\bibnamefont
  {Pande}},\ }\href {\doibase http://dx.doi.org/10.1063/1.3192309} {\bibfield
  {journal} {\bibinfo  {journal} {J. Chem. Phys.}\ }\textbf {\bibinfo {volume}
  {131}},\ \bibinfo {eid} {045106} (\bibinfo {year} {2009})}\BibitemShut
  {NoStop}%
\bibitem [{\citenamefont {Bowman}(2012)}]{bowman2012improved}%
  \BibitemOpen
  \bibfield  {author} {\bibinfo {author} {\bibfnamefont {G.~R.}\ \bibnamefont
  {Bowman}},\ }\href {\doibase http://dx.doi.org/10.1063/1.4755751} {\bibfield
  {journal} {\bibinfo  {journal} {J. Chem. Phys.}\ }\textbf {\bibinfo {volume}
  {137}},\ \bibinfo {eid} {134111} (\bibinfo {year} {2012})}\BibitemShut
  {NoStop}%
\bibitem [{\citenamefont {McGibbon}\ \emph
  {et~al.}(2014{\natexlab{a}})\citenamefont {McGibbon}, \citenamefont
  {Ramsundar}, \citenamefont {Sultan}, \citenamefont {Kiss},\ and\
  \citenamefont {Pande}}]{mcgibbon2013understanding}%
  \BibitemOpen
  \bibfield  {author} {\bibinfo {author} {\bibfnamefont {R.~T.}\ \bibnamefont
  {McGibbon}}, \bibinfo {author} {\bibfnamefont {B.}~\bibnamefont {Ramsundar}},
  \bibinfo {author} {\bibfnamefont {M.~M.}\ \bibnamefont {Sultan}}, \bibinfo
  {author} {\bibfnamefont {G.}~\bibnamefont {Kiss}}, \ and\ \bibinfo {author}
  {\bibfnamefont {V.~S.}\ \bibnamefont {Pande}},\ }in\ \href
  {http://jmlr.org/proceedings/papers/v32/mcgibbon14.html} {\emph {\bibinfo
  {booktitle} {Proceedings of the 31st International Conference on Machine
  Learning}}},\ Vol.~\bibinfo {volume} {32}\ (\bibinfo {address} {Beijing,
  China},\ \bibinfo {year} {2014})\ pp.\ \bibinfo {pages}
  {1197--1205}\BibitemShut {NoStop}%
\bibitem [{\citenamefont {McGibbon}\ \emph
  {et~al.}(2014{\natexlab{b}})\citenamefont {McGibbon}, \citenamefont
  {Beauchamp}, \citenamefont {Schwantes}, \citenamefont {Wang}, \citenamefont
  {Hern{\'a}ndez}, \citenamefont {Harrigan}, \citenamefont {Lane},
  \citenamefont {Swails},\ and\ \citenamefont {Pande}}]{McGibbon2014MDTraj}%
  \BibitemOpen
  \bibfield  {author} {\bibinfo {author} {\bibfnamefont {R.~T.}\ \bibnamefont
  {McGibbon}}, \bibinfo {author} {\bibfnamefont {K.~A.}\ \bibnamefont
  {Beauchamp}}, \bibinfo {author} {\bibfnamefont {C.~R.}\ \bibnamefont
  {Schwantes}}, \bibinfo {author} {\bibfnamefont {L.-P.}\ \bibnamefont {Wang}},
  \bibinfo {author} {\bibfnamefont {C.~X.}\ \bibnamefont {Hern{\'a}ndez}},
  \bibinfo {author} {\bibfnamefont {M.~P.}\ \bibnamefont {Harrigan}}, \bibinfo
  {author} {\bibfnamefont {T.~J.}\ \bibnamefont {Lane}}, \bibinfo {author}
  {\bibfnamefont {J.~M.}\ \bibnamefont {Swails}}, \ and\ \bibinfo {author}
  {\bibfnamefont {V.~S.}\ \bibnamefont {Pande}},\ }\href {\doibase
  10.1101/008896} {\bibfield  {journal} {\bibinfo  {journal} {bioRxiv}\ }
  (\bibinfo {year} {2014}{\natexlab{b}}),\ 10.1101/008896}\BibitemShut
  {NoStop}%
\bibitem [{\citenamefont {Snoek}, \citenamefont {Larochelle},\ and\
  \citenamefont {Adams}(2012)}]{NIPS2012_4522}%
  \BibitemOpen
  \bibfield  {author} {\bibinfo {author} {\bibfnamefont {J.}~\bibnamefont
  {Snoek}}, \bibinfo {author} {\bibfnamefont {H.}~\bibnamefont {Larochelle}}, \
  and\ \bibinfo {author} {\bibfnamefont {R.~P.}\ \bibnamefont {Adams}},\ }in\
  \href
  {http://papers.nips.cc/paper/4522-practical-bayesian-optimization-of-machine-learning-algorithms.pdf}
  {\emph {\bibinfo {booktitle} {Advances in Neural Information Processing
  Systems}}},\ Vol.~\bibinfo {volume} {25}\ (\bibinfo {address} {Lake Tahoe,
  USA},\ \bibinfo {year} {2012})\ pp.\ \bibinfo {pages}
  {2951--2959}\BibitemShut {NoStop}%
\bibitem [{\citenamefont {M{\"u}llner}(2013)}]{Mullner2013fastcluster}%
  \BibitemOpen
  \bibfield  {author} {\bibinfo {author} {\bibfnamefont {D.}~\bibnamefont
  {M{\"u}llner}},\ }\href {http://www.jstatsoft.org/v53/i09} {\bibfield
  {journal} {\bibinfo  {journal} {J. Stat. Softw.}\ }\textbf {\bibinfo {volume}
  {53}},\ \bibinfo {pages} {1} (\bibinfo {year} {2013})}\BibitemShut {NoStop}%
\bibitem [{\citenamefont {Hess}\ \emph {et~al.}(2008)\citenamefont {Hess},
  \citenamefont {Kutzner}, \citenamefont {van~der Spoel},\ and\ \citenamefont
  {Lindahl}}]{hess2008gromacs}%
  \BibitemOpen
  \bibfield  {author} {\bibinfo {author} {\bibfnamefont {B.}~\bibnamefont
  {Hess}}, \bibinfo {author} {\bibfnamefont {C.}~\bibnamefont {Kutzner}},
  \bibinfo {author} {\bibfnamefont {D.}~\bibnamefont {van~der Spoel}}, \ and\
  \bibinfo {author} {\bibfnamefont {E.}~\bibnamefont {Lindahl}},\ }\href
  {\doibase 10.1021/ct700301q} {\bibfield  {journal} {\bibinfo  {journal} {J.
  Chem. Theory Comput.}\ }\textbf {\bibinfo {volume} {4}},\ \bibinfo {pages}
  {435} (\bibinfo {year} {2008})}\BibitemShut {NoStop}%
\bibitem [{\citenamefont {Li}\ and\ \citenamefont
  {Br\"{u}schweiler}(2010)}]{liNMR2010}%
  \BibitemOpen
  \bibfield  {author} {\bibinfo {author} {\bibfnamefont {D.-W.}\ \bibnamefont
  {Li}}\ and\ \bibinfo {author} {\bibfnamefont {R.}~\bibnamefont
  {Br\"{u}schweiler}},\ }\href {\doibase 10.1002/ange.201001898} {\bibfield
  {journal} {\bibinfo  {journal} {Angew. Chem.}\ }\textbf {\bibinfo {volume}
  {122}},\ \bibinfo {pages} {6930} (\bibinfo {year} {2010})}\BibitemShut
  {NoStop}%
\bibitem [{\citenamefont {Jorgensen}\ \emph {et~al.}(1983)\citenamefont
  {Jorgensen}, \citenamefont {Chandrasekhar}, \citenamefont {Madura},
  \citenamefont {Impey},\ and\ \citenamefont
  {Klein}}]{jorgensen1983comparison}%
  \BibitemOpen
  \bibfield  {author} {\bibinfo {author} {\bibfnamefont {W.~L.}\ \bibnamefont
  {Jorgensen}}, \bibinfo {author} {\bibfnamefont {J.}~\bibnamefont
  {Chandrasekhar}}, \bibinfo {author} {\bibfnamefont {J.~D.}\ \bibnamefont
  {Madura}}, \bibinfo {author} {\bibfnamefont {R.~W.}\ \bibnamefont {Impey}}, \
  and\ \bibinfo {author} {\bibfnamefont {M.~L.}\ \bibnamefont {Klein}},\ }\href
  {\doibase http://dx.doi.org/10.1063/1.445869} {\bibfield  {journal} {\bibinfo
   {journal} {J. Chem. Phys.}\ }\textbf {\bibinfo {volume} {79}},\ \bibinfo
  {pages} {926} (\bibinfo {year} {1983})}\BibitemShut {NoStop}%
\bibitem [{\citenamefont {Bussi}, \citenamefont {Donadio},\ and\ \citenamefont
  {Parrinello}(2007)}]{bussi2007canonical}%
  \BibitemOpen
  \bibfield  {author} {\bibinfo {author} {\bibfnamefont {G.}~\bibnamefont
  {Bussi}}, \bibinfo {author} {\bibfnamefont {D.}~\bibnamefont {Donadio}}, \
  and\ \bibinfo {author} {\bibfnamefont {M.}~\bibnamefont {Parrinello}},\
  }\href {\doibase http://dx.doi.org/10.1063/1.2408420} {\bibfield  {journal}
  {\bibinfo  {journal} {J. Chem. Phys.}\ }\textbf {\bibinfo {volume} {126}},\
  \bibinfo {eid} {014101} (\bibinfo {year} {2007})}\BibitemShut {NoStop}%
\bibitem [{\citenamefont {Parrinello}\ and\ \citenamefont
  {Rahman}(1981)}]{parrinello1981polymorphic}%
  \BibitemOpen
  \bibfield  {author} {\bibinfo {author} {\bibfnamefont {M.}~\bibnamefont
  {Parrinello}}\ and\ \bibinfo {author} {\bibfnamefont {A.}~\bibnamefont
  {Rahman}},\ }\href {\doibase http://dx.doi.org/10.1063/1.328693} {\bibfield
  {journal} {\bibinfo  {journal} {J. Appl. Phys.}\ }\textbf {\bibinfo {volume}
  {52}},\ \bibinfo {pages} {7182} (\bibinfo {year} {1981})}\BibitemShut
  {NoStop}%
\bibitem [{\citenamefont {Darden}, \citenamefont {York},\ and\ \citenamefont
  {Pedersen}(1993)}]{darden1993particle}%
  \BibitemOpen
  \bibfield  {author} {\bibinfo {author} {\bibfnamefont {T.}~\bibnamefont
  {Darden}}, \bibinfo {author} {\bibfnamefont {D.}~\bibnamefont {York}}, \ and\
  \bibinfo {author} {\bibfnamefont {L.}~\bibnamefont {Pedersen}},\ }\href
  {\doibase http://dx.doi.org/10.1063/1.464397} {\bibfield  {journal} {\bibinfo
   {journal} {J. Chem. Phys.}\ }\textbf {\bibinfo {volume} {98}},\ \bibinfo
  {pages} {10089} (\bibinfo {year} {1993})}\BibitemShut {NoStop}%
\end{thebibliography}%
\end{document}